\documentclass[a4paper,11pt]{article} 
\usepackage[utf8]{inputenc}
\usepackage[english]{babel}
\usepackage[a4paper]{geometry}
\usepackage{amsmath}
\usepackage{amsthm}
\usepackage{amssymb}
\usepackage{color}
\usepackage{hyperref} 
\usepackage{enumerate}
\usepackage{bbm}
\usepackage{mathtools}

\def\bR{\mathbb{R}}

\def\bN{\mathbb{N}}

\def\bZ{\mathbb{Z}}
\def\cN{\mathcal{N}}
\def\cE{\mathcal{E}}
\def\cK{\mathcal{K}}
\def\cF{\mathcal{F}}
\def\cH{\mathcal{H}}

\def\cV{\mathcal{V}}
\def\cC{\mathcal{C}}
\def\cL{\mathcal{L}}
\def\cG{\mathcal{G}}

\def\cQ{\mathcal{Q}}

\def\cW{\mathcal{W}}

\def\be{\begin{equation*}}
\def\ee{\end{equation*}}

\def\aa{\mathfrak{a}}
\newcommand{\norm}[1]{\lVert #1 \rVert}
\def\la{\langle}
\def\ra{\rangle}

\def\hc{\text{h.c.}\ }

\newcommand{\abs}[1]{\lvert #1 \rvert}

\DeclareFontFamily{U}{mathx}{}
\DeclareFontShape{U}{mathx}{m}{n}{<-> mathx10}{}
\DeclareSymbolFont{mathx}{U}{mathx}{m}{n}
\DeclareMathAccent{\widehat}{0}{mathx}{"70}
\DeclareMathAccent{\widecheck}{0}{mathx}{"71}

\newtheorem{theorem}{Theorem}

\newtheorem{lemma}[theorem]{Lemma}

\title{A Note on the Construction of Trial States for the Dilute Bose Gas}

\author{
Morris Brooks\thanks{Institute for Mathematics, University of Zurich, Winterthurerstrasse 190, 8057 Zurich, Switzerland}
\and Jakob Oldenburg\thanks{CNRS \& CEREMADE, Universit\'e Paris--Dauphine, PSL University, 75016 Paris, France}
\and Diane Saint Aubin\footnotemark[1]
}

\begin{document}

\maketitle
\begin{abstract}
We review how the local particle number cutoff introduced in \cite{BOSS} is used to build trial states for the dilute Bose gas that capture the substantial correlation structure of the ground state in the thermodynamic limit.
In particular, we provide a simplified derivation of the Lee-Huang-Yang correction as an upper bound for the ground state energy.
\end{abstract}

\section{Introduction}
We consider a Bose gas in the thermodynamic limit, described by the Hamilton operator
\begin{equation*}\label{eq:HL} 
    H_{L,N} = \sum_{i=1}^N -\Delta_i + \sum_{1\leq i < j \leq N} V (x_i-x_j),
\end{equation*}
acting on the Hilbert space of permutation symmetric $L^2$-functions $L^2_s (\Lambda_L^N)$, where $\Lambda_L=[-L/2,L/2]^3$ is the thermodynamic box and we impose Dirichlet boundary conditions. In this article we study the dilute regime, i.e. we let the length of the box $L$ and the number of particles $N$ go jointly to infinity at a fixed but small value $\rho:=N/L^3$. We further assume the interaction  $0\leq V \in L^2 (\bR^3)$ to be radial and compactly supported. \\

It is well understood that correlations play a crucial role when it comes to low-energy properties of the operator $H_{L,N}$. In particular, it has been verified in \cite{Dy, LY} that the (ground state) energy per unit volume
\[ e(\rho) = \lim_{N, L \to \infty , \rho = N / L^3} \frac{E (N, L)}{L^3}  \]
is asymptotically given in terms of the scattering length as
\begin{align}
\label{Eq:Leading_Order}
   e(\rho)=4\pi \mathfrak{a}\rho^2+o_{\rho\rightarrow 0}(\rho^2).
\end{align}
Notably, for $V\neq 0$, the scattering length satisfies $8\pi\mathfrak{a}:=\int_{\mathbb R^3}f(x)V(x) dx<\int_{\mathbb R^3}V(x) dx$, where $f:\mathbb R^3\longrightarrow [0,1]$ denotes the solution of the zero-energy scattering equation 
\begin{equation}\label{eq:0en} \left[ -\Delta + \frac{1}{2} V \right] f = 0 \end{equation} 
with the boundary condition $f (x) \to 1$ as $|x| \to \infty$ and $\frac{1}{2}\int_{\mathbb R^3}V(x) dx$ is the optimal constant that can be obtained with a pure product state. Beyond the leading order asymptotics in \eqref{Eq:Leading_Order}, it has been predicted by Lee-Huang-Yang \cite{LHY}, based on  previous work of Bogoliubov \cite{Bo},
\begin{align}
\label{Eq:LHY}
     e(\rho)=4\pi \mathfrak{a}\rho^2
        \left(
        1+\frac{128}{15\sqrt{\pi}}\sqrt{\rho\mathfrak{a}^3}  \right)+o_{\rho\rightarrow 0} \! \left(\rho^{\frac{5}{2}}\right),
\end{align}
which has been verified at the Gross-Pitaevskii scale $\rho\sim N^{-2}$ in \cite{BBCS2} as well as in the thermodynamic limit by \cite{YY,FS1,FS2}, see also \cite{ESY,BCGOPS} for a second order upper bound. Following these landmark results, simplified proofs and improvements of \eqref{Eq:LHY} have been presented in e.g. \cite{Basti,BCGOPS,BCOPS,BCS,HST,FGJOM,B,BBCOS}. Extensions have been obtained in \cite{BSS2,NT} for more general trapping potentials, spectral properties beyond the ground state energy have been analyzed in intermediate regimes \cite{BCaS} and the free energy at low temperatures has been investigated in \cite{HHNST,HHST,FJGMOT}. Furthermore, the ground state and its condensation have been studied in \cite{ABS,BBCS1,BBCO,BSS1,LS,NNRT}.

The third-order correction to the asymptotics in \eqref{Eq:LHY}
\begin{align}
\label{Eq:Wu_Upper}
  e(\rho) = 4\pi \mathfrak{a}\rho^2
       \left(
        1+\frac{128}{15\sqrt{\pi}}\sqrt{\rho\mathfrak{a}^3} 
        + 8\left(\frac{4\pi}{3}-\sqrt{3}\right) \rho\mathfrak{a}^3 \log(\rho\mathfrak{a}^3) 
        + C \rho \frak{a}^3 
        \right)  
\end{align}
has been predicted by Wu \cite{Wu} in 1959, and then also by Hugenholtz-Pines \cite{HP} and Sawada \cite{Sa}, and has been verified at the Gross-Pitaevskii scale in \cite{COSS}. An upper bound matching \eqref{Eq:Wu_Upper} was recently established in the thermodynamic limit in \cite{BOSS}. In this manuscript, we review the method employed in \cite{BOSS}, and provide a simplified proof of the upper bound in \eqref{Eq:LHY} as the content of our main Theorem \ref{theorem:main}. Since, in contrast to the third order asymptotics \eqref{Eq:Wu_Upper}, only correlations on comparatively small length scales are relevant in \eqref{Eq:LHY}, our proof of Theorem \ref{theorem:main} focuses on the conceptual novelties in \cite{BOSS} and bypasses many technical complications.

\begin{theorem}\label{theorem:main}
Let $0\leq V\in L^2(\mathbb{R}^3)$ be radial and with compact support and let $\frak{a}$ be its scattering length. Then, there are constant $C,h > 0$ such that 
\begin{equation*} \label{eq:main} 
        e(\rho) \leq 4\pi \mathfrak{a}\rho^2
        \left(
        1+\frac{128}{15\sqrt{\pi}}\sqrt{\rho\mathfrak{a}^3} 
        + C \left(\rho \frak{a}^3\right)^{1/2+h}
        \right)
    \end{equation*}
for all $\rho \frak{a}^3 > 0$ small enough.  
\end{theorem}

The proof of Theorem \ref{theorem:main} follows immediately by the equivalences of ensembles, once we can provide a grand canonical trial state with the correct energy and particle density. By a well-known localization technique, it is furthermore sufficient to consider trial states on a periodic box of size $L=\rho^{-\gamma}$, given that $\gamma$ is large enough. In order to reach a precision that allows us to detect the Lee-Huang-Yang term, we require $\gamma>1$.

It is the content of Theorem \ref{theorem:Fockspaceresult} to provide an upper bound on the energy in this setting (grand canonical and periodic box of size $L=\rho^{-\gamma}$), by construction of an appropriate trial state. Since it can be a delicate task to find a state with the exact right number of particles, we first construct a trial state that has slightly less and one that has slightly more particles, which are then convex combined to a trial state with the correct density, see also \cite{BDS} for a similar argument. Before we state our main technical Theorem \ref{theorem:Fockspaceresult}, let us introduce the parameter $\kappa = (2\gamma-1)/(3\gamma-1)$, which allows us to express $L=N^{1-\kappa}$. By rescaling to the unit torus $\Lambda \equiv \Lambda_1 \simeq [-1/2 ;1/2]^3$, we observe that the grand-canonical operator $\bigoplus_{N=0}^\infty H_{L,N}$, acting on the Fock space $\cF (\Lambda) = \bigoplus_{N\geq 0}L^2_s(\Lambda_L^N)$, is unitarily equivalent to the operator $L^{-2}\cH_N$ acting on $\cF (\Lambda) = \bigoplus_{N\geq 0}L^2_s(\Lambda^N)$, where
\begin{equation*}\label{eq:HN-fock}
\cH_N  = \sum_{p\in\Lambda^*} p^2 a_p^* a_p
    + \frac{N^\kappa}{2N} \sum_{p,q,r\in\Lambda^*} \hat{V}(r/N^{1-\kappa}) a_{p+r}^* a_{q-r}^* a_q a_p,
\end{equation*}
and $a_p^*$ and $a_p$, for $p\in \Lambda^* := 2\pi\mathbb{Z}^3$, are the bosonic creation and annihilation operators.
\begin{theorem}\label{theorem:Fockspaceresult}
   Let $0\leq V\in L^2(\mathbb{R}^3)$ be radial and with compact support. Then, for any $\kappa\in (0,2/3)$, $\epsilon>0$ and $N$ large enough there exists $\Psi^{\pm}_N\in\cF (\Lambda)$ with $\norm{\Psi^\pm_N}=1$ and constants $C,h>0$, such that
    \begin{equation}\label{eq:psiN-N}
      {\la \Psi^-_N, \cN \Psi^-_N \ra + N^{\frac{3\kappa}{2}-\epsilon }\leq N \leq  \la \Psi^+_N, \cN \Psi^+_N \ra- N^{\frac{3\kappa}{2}-\epsilon }\leq 2N},
    \end{equation}
  and the energy of $\Psi^\pm_N$ is bounded by
    \begin{equation}\label{eq:mainN} 
    \begin{split}
        \la \Psi^\pm_N, \cH_N  \Psi^\pm_N \ra
        \leq \; &4\pi \mathfrak{a} N^{1+\kappa}
        +\frac{512\sqrt{\pi}}{15}\mathfrak{a}^{5/2}N^{5\kappa/2}
        + C N^{5\kappa/2-h}.
    \end{split}
    \end{equation}
\end{theorem}
With Theorem \ref{theorem:Fockspaceresult}, we can prove our main result. 
\begin{proof}[Proof of Theorem \ref{theorem:main}]
Rescaling to a box of length $L_0=\rho^{-\gamma}$, with $\gamma:=\frac{1-\kappa}{2-3\kappa}$, we obtain states $\Psi^\pm_{L_0}$ on the periodic box $\Lambda_{L_0}$ satisfying
\begin{align}
\label{eq:particle_L_0}
    L_0^{-3}\la \Psi^-_{L_0}, \cN \Psi^-_{L_0} \ra + \rho^{3/2+\frac{\epsilon}{2-3\kappa}}& \leq \rho \leq  L_0^{-3}\la \Psi^+_{L_0}, \cN \Psi^+_{L_0} \ra- \rho^{3/2+\frac{\epsilon}{2-3\kappa}} \leq 2 \rho ,\\
    \nonumber 
    L_0^{-3}\la \Psi^\pm_{L_0}, \cH_N  \Psi^\pm_{L_0} \ra & \leq 4\pi \mathfrak{a} \rho^2  +\frac{512\sqrt{\pi}}{15}\mathfrak{a}^{5/2}\rho^{5/2}  + C \rho^{5/2+\frac{h}{2-3\kappa}},
\end{align}
Utilizing the methods in \cite[Appendix A]{BCS}, we find states $\Phi^{\pm}_L$ on the thermodynamic box $\Lambda_L$ with 
\begin{align}
\label{eq:particle_L}
   |L^{-3}\la \Phi^\pm_{L}, \cN \Phi^\pm_{L} \ra-L_0^{-3}\la \Psi^\pm_{L_0}, \cN \Psi^\pm_{L_0} \ra|\leq \rho^{3/2+\frac{\epsilon}{2-3\kappa}},
\end{align}
and $L^{-3}\la \Phi^\pm_{L}, \cH_N  \Phi^\pm_{L} \ra  \leq L_0^{-3}\la \Psi^\pm_{L_0}, \cH_N  \Psi^\pm_{L_0} \ra + C \rho^{1/2+2\gamma-\frac{\epsilon}{2-3\kappa}}$. Consequently, \eqref{eq:particle_L_0} and \eqref{eq:particle_L} guarantee that there exists a convex combination $\rho L^3=t \la \Phi^-_{L}, \cN \Phi^-_{L} \ra + (1-t)\la \Phi^+_{L}, \cN \Phi^+_{L} \ra$, $t\in [0,1]$, or put in a different way there exists a $t\in [0,1]$ such that
\begin{align*}
\Gamma_L:=t\Pi_{\Phi^-_L}+(1-t)\Pi_{\Phi^+_L}    
\end{align*}
 defined via projections $\Pi_{\Phi^\pm_L}$ satisfies $\mathrm{Tr}[\cN \, \Gamma_L]=\rho L^3$. According to the equivalences of ensembles \cite{Rue} we conclude
 \begin{align*}
    e(\rho)\leq \liminf_L L^{-3}\mathrm{Tr}[\cH_L \Gamma_L]\leq 4\pi \mathfrak{a} \rho^2  +\frac{512\sqrt{\pi}}{15}\mathfrak{a}^{5/2}\rho^{5/2}   + C \rho^{5/2+\frac{h}{2-3\kappa}}+ C\rho^{1/2+2\gamma-\frac{\epsilon}{2-3\kappa}},
\end{align*}
where we choose e.g. $\epsilon:=1/3$ and $\kappa<\frac{2}{3}$ large enough such that $1/2+2\gamma-\frac{1}{3(2-3\kappa)}$ is larger than $\frac{5}{2}$.
\end{proof} 

\medskip

\textit{Acknowledgements.} We gratefully acknowledge financial support from the Swiss National Science Foundation through the Grant “Bogoliubov theory for bosonic systems” and from the European Research Council through the ERC-AdG CLaQS. Morris Brooks gratefully acknowledges funding from the UZH Postdoc Grant FK-25-103.
Jakob Oldenburg was supported by the French State support managed by ANR under the France 2030 program through the MaQuI CNRS Risky and High-Impact Research programme $(\textrm{RI})^2$ (grant agreement ANR-24-RRII-0001).

\section{Construction and properties of the trial state}

We sketch the construction of the trial state $\Psi_N \in \cF (\Lambda )$ satisfying \eqref{eq:psiN-N} and \eqref{eq:mainN}. 

It is well-understood that the correct leading order $4\pi \mathfrak{a}N^{1+\kappa}$ in the ground-state energy of $\cH_N$ arises from a perfect Bose-Einstein condensate dressed by correlations between particle pairs on the short length-scale $N^{\kappa-1}$, which is proportional to the range of the potential. While a perfect condensate is simply a product function in the canonical picture, there are more ambiguities in the grand-canonical setting and we decide to model the condensate by a coherent state, i.e. we use the Weyl operator $W_{N_0}$ satisfying \begin{equation}\label{eq:weyl1} W_{N_0}^* a_p W_{N_0} = a_p + \delta_{p,0} \sqrt{N_0} \end{equation}
for all $p \in \Lambda^* = 2\pi \bZ^3$, which creates (in average) $N_0$ particles with zero momentum. The short-range correlations are then implemented using the Bogoliubov transformation
\begin{align}
\label{eq:BT0}
    e^{B-B^*}:=e^{ \frac{1}{2} \sum_{p \in \Lambda^* \backslash \{ 0 \}} \mu_p a_p a_{-p}-\left( \frac{1}{2} \sum_{p \in \Lambda^* \backslash \{ 0 \}} \mu_p a_p a_{-p}\right)^*},
\end{align}
with sensibly chosen coefficients $\mu_p$, which actually creates excitations on all length-scales up to $N^{-\kappa/2}$. We will then consider a trial state of the form 
\begin{equation}\label{eq:psi0} \Psi_N = W_{N_0} e^{B-B^*} \xi, \end{equation} 
where the state $\xi$ encodes additional correlations on the healing length $N^{-\kappa/2}$ that contribute to the Lee-Huang-Yang term. By a formal Duhamel expansion, it is evident that the excitation vector should be of the form
\begin{align}
\label{Eq:Pre_Zeta}
    \xi_\mathrm{formal}:=e^{A^* - A} \Omega
\end{align}
with $\Omega$ being the vacuum state and $A$ given as the operator 
\begin{equation*}\label{eq:A-def}
A = \frac{1}{\sqrt{N}}\sum_{p,q \in \Lambda^*} \eta_p\sigma_q \, a_{p+q} a_{-p} a_{-q},
\end{equation*} 
where the correct coefficients $\eta_p$ and $\sigma_q$ are introduced in the following Subsection \ref{subsec:bogo}. The formal definition in \eqref{Eq:Pre_Zeta} suffers from multiple issues. First of all since $A$ is cubic in creation and annihilation operators we expect that $A^*-A$ is not essentially selfadjoint on the finite-particle subspace, hence there might be an ambiguity in the choice of a selfadjoint extension. Even more important, and related to the first point, it is not possible to (formally) close a Grönwall argument, which would be necessary to control the remainder in a Duhamel expansion. \\

Both issues with $\xi_\mathrm{formal}$ are addressed in this manuscript, using a cutoff operator $\Theta$ that allows us to control the local particle density. With $A$ and $\Theta$ at hand, we introduce the excitation vector $\xi$ in the definition of our trial state $\Psi_N$ in \eqref{eq:psi0} as
\begin{align*}
    \xi:=e^{\Theta A^* - A\Theta} \Omega.
\end{align*}
Notably, we choose $\Theta$ such that $A$ is bounded on the range of $\Theta$ and in particular we observe that the unique anti-selfadjoint extension of $\Theta A^* - A\Theta$ is given by $(A\Theta)^*-(A\Theta)$.
The relevant contributions of $\xi$ to the energy are then extracted via Duhamel expansion of order up to two.

\subsection{Coefficients} 
\label{subsec:bogo}
The coefficients appearing in the Bogoliubov transformation \eqref{eq:BT0} are chosen via the scattering problem \eqref{eq:0en} as follows: define first for $p\in \Lambda^*_+ = 2\pi \mathbb{Z}^3\backslash\{0\}$
\begin{equation}\label{eq:defetainfty}
\eta_{\infty,p} = \frac{-N^\kappa \widehat{Vf}(p/N^{1-\kappa})}{2p^2}
\end{equation}
and let $\eta_{\infty,0}= 0$.
Note that the solution of \eqref{eq:0en} can be rewritten as $1-f = (-2\Delta)^{-1}Vf$ emphasizing that these coefficients are the Fourier coefficients of a rescaled version of $1-f$. However, $1-f$ only decays as $\frac{\aa}{\abs{x}}$ on $\mathbb{R}^3$ and hence is not a function on the torus. The torus implicitly imposes a long-distance cutoff and thereby (with the choice $\eta_{\infty,0}= 0$) also regularizes the scattering solution in the sense that $\eta_\infty \in L^1\cap L^\infty$ whereas this is not true for $1-f$. We show in Lemma \ref{lemma:scattering} below that $\eta_\infty$ satisfies a discrete version of the scattering equation with a subleading finite size error.
Next, define for $p\in \Lambda^*$
\begin{equation*}\label{eq:defmuinfty}
\mu_{\infty,p} = -\frac{1}{4} \log (1-4\eta_{\infty,p}).
\end{equation*}
Observe \[ \tanh (2\mu_{\infty,p})  = - \frac{
\widehat{Vf}(p/N^{1-\kappa})}{p^2 + \widehat{Vf}(p/N^{1-\kappa})},\]
and hence, these coefficients would be the right choice to diagonalize a quadratic Hamiltonian consisting of kinetic energy and an interaction $Vf$.
As the action of a Bogoliubov transformation contains the $\sinh$ and $\cosh$ of the kernels we define
\[
\sigma_{\infty,p} = \sinh \mu_{\infty,p} \textrm{ and } \gamma_{\infty,p} = \cosh \mu_{\infty,p}.
\]

We do not directly use these coefficients but rather impose a momentum cutoff as it turns out that momenta $\abs{p}\ll N^{\kappa/2}$ corresponding to the healing length do not matter to the precision we require. Pick $\epsilon > 0$ and define the length scale $\ell_\sigma = N^{-\kappa/2+\epsilon}$ slightly bigger than the healing length. We will remove all momenta below $\ell_\sigma^{-1}$. Note that this implies that the correlations we implement live on short scales as the Fourier transforms of the various kernels should decay rapidly in position space at distances much larger than $\ell_\sigma$. This decay will be crucial when analyzing the cubic transformation. In order to make this statement rigorous, we use that the kernels in momentum space can be extended to $\mathbb{R}^3$ in an obvious manner and use a smooth cutoff as we need to take derivatives in momentum space in order to prove the decay estimates.
Thus, first define for every $0\neq p \in \bR^3$
\[
\begin{split}
&\eta_\infty (p) = - \frac{N^\kappa\widehat{Vf} (p/ N^{1-\kappa})}{2 p^2},
 \quad\mu_\infty (p) = -\frac{1}{4}  \log (1-4 \eta_\infty (p)) ,\\
 &\sigma_\infty (p) = \sinh \mu_\infty (p),
 \quad\gamma_\infty (p) = \cosh \mu_\infty (p).
 \end{split}
\]
Now we introduce the low momentum cutoff: Pick $\chi_l \in C^{\infty}(\mathbb{R})$ with $0\leq \chi_l \leq 1$, $\chi_l(x)=0$ if $x\leq 1$, $\chi_l(x)=1$ if $x\geq 2$, and we set $\chi (p) = \chi_l (\ell_\sigma\abs{p})$ and $\tilde\chi (p) = \chi_l (\ell_\eta\abs{p})$ for $\ell_\eta = N^{-1+\kappa+\epsilon} \ll \ell_\sigma$.
We can now define the kernels with cutoff as
\begin{align}
\label{eq:definition_cut_off_Bog_objects}
     \mu (p) \coloneqq \mu_{\infty}(p) \chi (p),
    \quad \eta(p) = \eta_{\infty}(p)\tilde\chi(p),
    \quad \sigma (p) = \sinh \mu (p),
    \quad \gamma (p) = \cosh \mu (p).
\end{align}
For all lattice points $p \in \Lambda^*$ we write
$\mu_p=\mu(p)$, $\sigma_p = \sigma(p)$, $\gamma_p = \gamma(p)$, $\eta_p = \eta(p)$. In position space, we set 
\[ \check{\sigma} (x) = \sum_{p \in \Lambda^*} \sigma_p e^{i p \cdot x} , \quad
\widecheck{\gamma-1} (x) = \sum_{p\in\Lambda^*} (\gamma_p-1) e^{i p \cdot x},
\quad
\text{and}\quad
\check{\eta} (x) = \sum_{p \in \Lambda^*} \eta_p e^{i p \cdot x}. \]

The discrete scattering equation and some further useful properties of $\eta_{\infty,p}$, $\sigma_{\infty,p}$ and $\gamma_{\infty,p}$ are contained in Lemma \ref{lemma:scattering}.

 \begin{lemma}\label{lemma:scattering}
 The coefficients $\eta_{\infty,p}$ satisfy the discrete scattering equation
\begin{equation}\label{eq:scatteringdiscrete}
\begin{split}
p^2 \eta_{\infty,p} +\frac{N^\kappa}{2} \hat{V}(p/N^{1-\kappa}) +\frac{N^\kappa}{2N} \sum_{q\in \Lambda_+^*} \hat{V}((p-q)/N^{1-\kappa})\eta_{\infty,q}
= O(N^{2\kappa-1})
\\
\frac{N^\kappa}{2N} \sum_{q\in \Lambda_+^*} \hat{V}(\frac{p-q}{N^{1-\kappa}})\eta_{\infty,q}
= - \frac{N^\kappa}{2} \widehat{Vw}(\frac{p}{N^{1-\kappa}}) + O(N^{2\kappa-1})
\end{split}
\end{equation} 
uniformly in $p \in \Lambda^*_+$. Furthermore, 
\begin{equation}\label{eq:etacorrection}
\sum_{p\in \Lambda_+^*} N^\kappa \hat{V}(p/N^{1-\kappa})\eta_{\infty,p} 
= (8 \pi \mathfrak{a}- \hat{V}(0)) N^{1+\kappa} + O(N^{2\kappa}).
\end{equation} 
Additionally, $\gamma_{\infty,p}$ and $\sigma_{\infty,p}$ satisfy the bounds
 \begin{equation}\label{eq:sigmainftymomentum}
 \abs{\gamma_{\infty,p}-1}\leq \abs{\sigma_{\infty,p}}
 \leq C \abs{\widehat{Vf}(p/N^{1-\kappa})} \,  \min \Big( \frac{N^{\kappa/4}}{\abs{p}^{1/2}}, \frac{N^{\kappa}}{\abs{p}^{2}} \Big),
 \end{equation}
 and therefore,
 \begin{equation*}\label{eq:sigmainftynorms}
 \norm{\sigma_\infty}_2^2,  \norm{\gamma_\infty-1}_2^2
 \leq C N^{3\kappa/2},
 \quad
 \norm{\sigma_\infty}_1,  \norm{\gamma_\infty-1}_1
 \leq CN.
 \end{equation*}
Lastly, we have
 \begin{equation}\label{eq:gammasigmainfty-etainftysum}
 \norm{\gamma_{\infty}\sigma_{\infty} - \eta_{\infty}}_1,  \norm{\gamma_\infty-1}_1
 \leq CN^{3\kappa/2}.
 \end{equation}
 
 \end{lemma}
The proof can be found in \cite[Lemma 3]{BOSS}, we also include a simplified version in the appendix for completeness. 
In the following we show properties of the kernels containing the momentum cutoff and in particular make the decay estimates in position space rigorous.
The proof is a straightforward adaptation of \cite[Lemma 5]{BOSS}, we give the details in the appendix. 
\begin{lemma}\label{lemma:propertiesquadratickernels}
The kernels $\gamma, \sigma$ and $\eta$ satisfy
\begin{equation*}\label{eq:boundsigmamomentum}
\begin{split}
    \abs{\gamma (p) -1}
    \leq\abs{\sigma (p) }
    &\leq C \abs{\widehat{Vf}(p/ N^{1-\kappa})} \min(N^{\kappa/4}\abs{p}^{-1/2}, N^\kappa  \abs{p}^{-2})\\
    \abs{\eta (p) }
    &\leq C \abs{\widehat{Vf}(p/ N^{1-\kappa})} N^\kappa  \abs{p}^{-2},
    \end{split}
\end{equation*}
and thus also
\begin{equation*}
    \norm{\sigma}_\infty, \norm{\gamma-1}_\infty\leq CN^{\epsilon/2},
    \quad
    \norm{\sigma}_2^2,  \norm{\gamma-1}_2^2 
    \leq C N^{3\kappa/2},
    \quad
    \norm{\sigma}_1,  \norm{\gamma-1}_1 
    \leq C N
\end{equation*}

\begin{equation}\label{eq:normseta}
    \norm{\eta}_\infty\leq C N^{3\kappa-2+2\epsilon}, \quad \norm{\eta}_2^2\leq C N^{3\kappa-1+\epsilon}, \quad \norm{\eta}_1\leq CN.
\end{equation}

Moreover 
\begin{equation}\label{eq:eta-etainftysum}
	\norm{\eta-\eta_\infty}_1 \leq C N^{1-\epsilon},
    \norm{\sigma -\sigma_\infty}_1\leq N^{3\kappa/2-\epsilon}.
\end{equation}

Furthermore, for any $m\geq 1$, there exists $C_m > 0$ such that 
\begin{equation}\label{eq:boundsigmaxm}
\abs{\check\sigma (x)}, \abs{\widecheck{\gamma-1} (x)}
\leq C_m  N
\left(\frac{\ell_\sigma}{\abs{x}}\right)^m
\end{equation}
where $\ell_\sigma=N^{-\kappa/2+\epsilon}$.
The previous bounds also imply
\begin{equation}\label{eq:boundsigmax1}
\norm{\check\sigma}_1, \norm{\widecheck{\gamma-1}}_1
\leq C N^{2\epsilon},
\quad \norm{\check\sigma}_\infty, \norm{\widecheck{\gamma-1}}_\infty
\leq C N
\end{equation}
Similarly, with $\ell_\eta=N^{-(1-\kappa)+\epsilon}$, we have for any $m\geq 1$
\begin{equation}\label{eq:boundetam}
\abs{\check\eta (x)}
\leq C_m  N
\left(\frac{\ell_\eta}{\abs{x}}\right)^m.
\end{equation}
In addition we have for $r>0$
\begin{align}
\label{eq:boundnablasigma}
    \int_{|x|\geq r}dx | \nabla \check  \sigma(x)|^2  \leq C_m N^{1+\kappa+2\epsilon}\left(\frac{\ell_\sigma}{r}\right)^{2m}.
\end{align}
\end{lemma}

\subsection{The Local Particle Number Cutoff}

We now construct the local particle number cutoff $\Theta$ used in the trial state in \eqref{eq:psi0}. We impose conditions on the number of excitations on balls of size $\ell_{B} = N^{-\kappa/2+2\epsilon}$ with $\epsilon>0$ small enough as before, and in particular we will assume $0<\epsilon<\frac{2-3\kappa}{7}$ to ensure that the expected number of excitations in a ball of radius $\ell_{B} $ is of the order $O(1)$. As we will see, on the support of the cutoff the number of excitations in such balls is bounded by $N^\delta$ for a suitably chosen small $\delta > 0$.
Observe that $\ell_\eta \ll \ell_\sigma \ll \ell_B$, and hence, the transformations create correlations on scales much shorter than the ones on which restrict the particle number. This allows us to estimate error terms as follows: at short scales we use the local particle number cutoff, whereas at large scales, we use the rapid decay of the kernels.

In order to construct $\Theta$, let $\chi_{\ell_B} \in C_c^\infty(\mathbb{R}^3)$ be such that $0\leq \chi_{\ell_B}\leq 1$, $\chi_{\ell_B}(x)=1$ if $\abs{x}\leq \ell_B$ and $\chi_{\ell_B}(x)=0$ if $\abs{x}\geq 2\ell_B$. We then define
\begin{equation*}
    \cN_{\chi_{\ell_B}} (w) \coloneqq \int dv \, \chi_{\ell_B}(w-v) a_v^* a_v,
\end{equation*}
which essentially counts the number of particle close to a given point $w$, and for $n \in \bN$ we introduce 
\begin{equation*} \label{eq:def-Theta} 
    \cL_n \coloneqq \int dw \, (\cN_{\chi_{\ell_B}} (w) +1)^n.
\end{equation*}
The cut-off operator is then defined as
$
    \Theta(\cL_n/N^{\frac{1}{\delta}}) \, ,
$
where $\Theta\in C^\infty(\mathbb{R})$ is a smooth function with $0\leq \Theta\leq1$, $\Theta(x)=1$ if $x\leq 1$ and $\Theta(x)=0$ if $x\geq 2$. Note that in the following, we choose $n=n(\delta)$ large enough such that on the support of the cutoff operator $\Theta(\cL_n/N^{\frac{1}{\delta}})$ the value of the local particle number operators $\cN_{\chi_\ell} (w)$ is of the magnitude $N^\delta$. To be precise, we fix $n(\delta):=\lceil \frac{1}{\delta}+\frac{1}{\delta^2}\rceil$, which is sufficiently large as is explained in the proof of Lemma \ref{lemma:propertiesL_pre}.

With the cutoff $\Theta(\cL_n/N^{\frac{1}{\delta}})$ at hand, we introduce the unitary cubic transformation \begin{equation}\label{eq:defcubictransform}
    \exp \Big(\Theta(\cL_n/N^{\frac{1}{\delta}})\int dxdydz \check{\eta}(x-y)\check{\sigma}(x-z) a_x^* a_y^* a_z^* -\hc\Big)
    \eqqcolon \exp(\Theta A^*-\hc),
\end{equation}
where, on the right side of \eqref{eq:defcubictransform}, we have omitted the argument $\cL_n/N^{\frac{1}{\delta}}$ in the operator $\Theta(\cL_n/N^{\frac{1}{\delta}})$ to keep the notation light. While we deem the inclusion of the cut-off $\Theta$ a necessity to obtain a convergent Duhamel expansion (which we will truncate at a finite level), it is a pivotal observation that the trial state $\xi:=\exp(A^*\Theta-\hc)\Omega$ is mostly supported on the subspace $\Theta=1$; the difference is of order $N^{-1/(2\delta)}$, see \eqref{eq:norm1-Theta}. Clearly, such a property is relevant as we want to recover formal computations based on the (ill-defined) trial state $\exp(A^*-A)\Omega$. In particular, the support properties of $\xi$ allow us to neglect commutators with the cut-off $\Theta$ arising in a Duhamel expansion, which could easily spoil the energy of our trial state. In order to derive such support properties, we first show that the expectation of $\mathcal{L}_n$ in $\xi$ is finite as part of the fundamental Lemma \ref{lemma:propertiesL}. Before we show Lemma \ref{lemma:propertiesL}, we collect various properties of $\mathcal{L}_n$ and the local particle number operators in the subsequent Lemma \ref{lemma:propertiesL_pre}. We find it convenient to further introduce for $\ell>0$ a local particle number operator with a sharp cut-off
\begin{equation*}
    \cN_{\ell} (w) \coloneqq \int dv \mathbbm{1}(\abs{w-v}\leq \ell) \, a_v^* a_v .
\end{equation*}

\begin{lemma}\label{lemma:propertiesL_pre}
For every $k \in \bN$ and state $\zeta \in \cF (\Lambda )$, there is a $C_k > 0$ such that
\begin{align}
\label{eq:pull_through_local}
        C_k^{-1}\norm{\prod_{i=1}^k a_{x_i} (\cN_{\chi_{\ell_B}} (w) +1)\zeta}
    \leq \norm{ (\cN_{\chi_{\ell_B}} (w) +1)\prod_{i=1}^k a_{x_i}\zeta}
    \leq \norm{\prod_{i=1}^k a_{x_i} (\cN_{\chi_{\ell_B}} (w) +1)\zeta}.
\end{align}
Moreover, there is a $C > 0$ such that, for every $\delta >0$ 
{
\begin{equation}\label{eq:localNbound}
    \cN_{k \ell_B} (w)\Theta(\cL_n/N^{\frac{1}{\delta}})
    \leq 
   C  k^3 N^{\delta} \Theta( \cL_n/N^{\frac{1}{\delta}})
\end{equation}
and also 
\begin{equation}\label{eq:localNboundexample}
    \norm{\cN_{k\ell_B}^{1/2}(w)\prod_{j=1}^m a_{x_j}\Theta(\cL_n/N^{\frac{1}{\delta}})\zeta}^2
    \leq C  k^3 N^{\delta}
    \norm{\prod_{j=1}^m a_{x_j}\Theta(\cL_n/N^{\frac{1}{\delta}})\zeta}^2 
\end{equation}
}
for all $w \in \Lambda$, $k\geq 1$ and $m \in \bN$. 
\end{lemma}
\begin{proof}
We note that
$
    [\cN_{\chi_{\ell_B}} (w), a_x^*] = \chi_{\ell_B}(w-x)a_x^*,
$
which immediately gives us \eqref{eq:pull_through_local}. Regarding \eqref{eq:localNbound}, we have that $(\cN_{\ell_B/2} (w)+1)^n$ is bounded by $ (\cN_{\chi_{\ell_B}} (y) +1)^n$ for arbitrary $w$ and $y$ with $\abs{w-y}\leq \ell_B/2$, and as a consequence we have for all $w\in \Lambda$
\begin{equation*}
\begin{split} \cL_n   &= \int dy \, (\cN_{\chi_{\ell_B}} (y) +1)^n \geq \int_{|w-y| \leq \ell_B/2}  dy \, (\cN_{\chi_{\ell_B}} (y) +1)^n \\ &\geq \int_{|w-y| \leq \ell_B/2}  dy \, (\cN_{\ell_B/2} (w) +1)^n  \geq C \ell_B^3   (\cN_{\ell_B/2} (w) +1)^n.
\end{split}  
\end{equation*}
Since the ball of radius $\ell_B$ can be covered by finitely many balls of radius $\ell_B/2$, i.e. $B_{\ell_B}(w)\subseteq \bigcup_{i\in I}B_{\ell_B/2}(w_i)$ with $|I|<\infty$, we obtain on the support of $\Theta(\cL_n/N^{\frac{1}{\delta}})$
\begin{equation*}
    \cN_{\ell_B} (w) \leq \sum_{i\in I}\cN_{\ell_B/2} (w_i)
    \leq C \ell_B^{-3/n}N^{1/(n\delta)} 
    \leq C  N^{\delta}
\end{equation*}
where we have used that $n= \lceil \frac{1}{\delta}+\frac{1}{\delta^2}\rceil \geq \frac{3\kappa}{2\delta}+\frac{1}{\delta^2}$. 
Decomposing balls of radius $k \ell_B$ into balls of radius $\ell_B$, we find 
\begin{equation*}
    \cN_{k\ell_B} (w) 
    \leq C  k^3 N^{\delta}.
\end{equation*}
if $n$ is large enough. The final statement \eqref{eq:localNboundexample} can be shown analogously.
\end{proof}

With the auxiliary Lemma \ref{lemma:propertiesL_pre} at hand, we can control the expectation value of $\mathcal{L}_n$ with respect to our trial state $\exp( {\Theta A^*}-\hc)\Omega$ in Lemma \ref{lemma:propertiesL}. In particular, we obtain strong bounds on the probability of leaving the spectral subspace $\Theta=1$ by Markov's inequality, see \eqref{eq:norm1-Theta}.

\begin{lemma}\label{lemma:propertiesL}
    For $\delta > 0$ small enough and $0\leq s\leq 1$, we have
    \begin{equation}\label{eq:expectationL}
    \la \exp(s  {\Theta A^*}-\hc)\Omega, \cL_n \exp(s  {\Theta A^*}-\hc)\Omega\ra \leq C.
\end{equation}
 As a consequence
\begin{equation}\label{eq:norm1-Theta}
    \norm{(1-\Theta(\cL_n/N^{\frac{1}{\delta}}))\exp( {\Theta A^*}-\hc)\Omega} \leq C N^{-1/(2\delta)}.
\end{equation}
Also, for $m\in \mathbb N$ and $\delta > 0$ small enough, we have 
 {
\begin{equation}\label{eq:easyboundNm}
    \la \exp( {\Theta A^*}-\hc ) \Omega , \cN^m \exp ( {\Theta A^*}- \hc ) \Omega \ra
    \leq C N^{\frac{3\kappa m}{2}-6\epsilon m}.
\end{equation}
}
Finally, we observe that, for any $k$ and $\nu$, we have  
  \begin{equation}\label{eq:AThetaA}
    \la \exp( {\Theta A^*}-\hc)\Omega,  {A}  \mathbbm{1}_{[1;\infty)} (\cL_n / N^{\frac{1}{\delta}})  { \mathcal{N}^k A^*} \exp( {\Theta A^*}-\hc)\Omega\ra \leq C N^{-\nu}
\end{equation}
if $\delta > 0$ is chosen small enough. The same holds for $\mathbbm{1}_{[1;\infty)}$ being replaced by $1-\Theta$.
\end{lemma}
\begin{proof}
We note that \eqref{eq:norm1-Theta} is an immediate consequence of (\ref{eq:expectationL}), and \eqref{eq:easyboundNm} follows from (\ref{eq:expectationL}) for $m\leq n=\lceil \frac{1}{\delta}+\frac{1}{\delta^2}\rceil$ together with the simple observation
\begin{align}
\label{eq:compare_N_with_L}
    \cN^m \leq \left(C \ell_B^{-3} \int dw \, \cN_{\chi_{\ell_B}} (w)\right)^m \leq C^m \ell_B^{-3m} \cL^{m/n}_n\leq C^m \ell_B^{-3m}\cL_n.
\end{align}

To prove (\ref{eq:expectationL}), we set $\xi=\exp( {t\Theta A^*}-\hc)\Omega$ and compute
\begin{equation*}
\begin{split}
 & \Big| \frac{d}{dt} \langle e^{t   {\Theta A^*}-\hc } \Omega, \cL_n     e^{t   {\Theta A^*}-\hc} \Omega\rangle \Big|= 2\Big|  \mathfrak{Re}\langle \xi, [\cL_n,\Theta A^*]    \xi\rangle \Big|\\
    &\leq \abs{\la\xi, [\cL_n, \frac{1}{\sqrt{N}}  {\Theta(\cL_n/N^{\frac{1}{\delta}})} \int dx_1 dx_2 dx_3 \check{\eta}(x_1-x_2)\check{\sigma}(x_1-x_3)  a_{x_1}^* a_{x_2}^* a_{x_3}^* ]\xi\ra}\\
    &\leq \frac{1}{\sqrt{N}} \int dx_1 dx_2 dx_3 dw \, |\check{\eta}(x_1-x_2)\check{\sigma}(x_1-x_3)|   \, \abs{\la\xi, {\Theta(\cL_n/N^{\frac{1}{\delta}})}[(\cN_{\chi_{\ell_B}} (w) +1)^n, a_{x_1}^* a_{x_2}^* a_{x_3}^*] \xi\ra}\\
    &\leq \sum_{k=0}^{n-1}\int dx_1 dx_2 dx_3dw  \,  {\frac{1}{\sqrt{N}}}|\check{\eta}(x_1-x_2)\check{\sigma}(x_1-x_3)| \\ &\hspace{.4cm} \times \big| \la\xi, {\Theta(\cL_n/N^{\frac{1}{\delta}})}(\cN_{\chi_{\ell_B}} (w) +1)^k 
    \sum_{i=1}^3\chi_{\ell_B} (x_i-w)  a_{x_1}^* a_{x_2}^* a_{x_3}^*
    (\cN_{\chi_{\ell_B}} (w) +1)^{n-k-1} \xi\ra \big| \\
    &\leq C
    \sum_{i=1}^3
    \int dx_1 dx_2 dx_3 dw  {\frac{1}{\sqrt{N}}}|\check{\eta}(x_1-x_2)\check{\sigma}(x_1-x_3)|
    \chi_{\ell_B}(x_i-w)\\
    &\hspace{.4cm} \times\norm{a_{x_1} a_{x_2} a_{x_3}   (\cN_{\chi_{\ell_B}} (w) +1)^{(n-1)/2}  {\Theta(\cL_n/N^{\frac{1}{\delta}})}\xi}
    \norm{(\cN_{\chi_{\ell_B}} (w) +1)^{(n-1)/2} \xi}\\
    & =: \sum_{i=1}^3\left(X^{(i)}_{\ell_B}+X^{(i)}_\infty\right),
\end{split}
\end{equation*}
 {where we split the kernel $\frac{1}{\sqrt{N}}\check{\eta}(u)\check{\sigma}(v)=\check{A}_{\ell_B}(u,v)+\check{A}_\infty(u,v)$ in the main contribution $\check{A}_{\ell_B}(u,v):=\mathbbm{1}(|u|,|v|\leq \ell_B)\check{A}(u,v)$ and the residuum $\check{A}_\infty$}. Considering e.g. $i=1$, Cauchy-Schwarz yields
\[ \begin{split}
    &X^{(1)}_{\ell_B}\leq C
    \sum_{i=1}^3
    \Big(\int dx_1 dx_2 dx_3 dw \, \chi_{\ell_B}(x_1-w)\mathbbm{1}(\abs{x_1-x_2},\abs{x_1-x_3}\leq \ell_B)\\
    &\hspace{3cm}\times\norm{a_{x_1} a_{x_2} a_{x_3}   (\cN_{\chi_{\ell_B}} (w) +1)^{\frac{(n-1)}{2}} {\Theta(\cL_n/N^{\frac{1}{\delta}})} \xi}^2 \Big)^{1/2}\\
    &\hspace{.1cm} \times  \Big(\int dx_1 dx_2 dx_3 dw |\check{A}_{\ell_B}(x_1-x_2,x_1-x_3)|^2 
    \chi_{\ell_B}(x_1-w)\norm{(\cN_{\chi_{\ell_B}} (w) +1)^{\frac{(n-1)}{2}} \xi}^2\Big)^{1/2} \\
    & \leq C  \Big(\int  dw \, \norm{\cN_{2\ell_B} (w)(\cN_{\chi_{\ell_B}} (w) +1)^{\frac{n}{2}}  {\Theta(\cL_n/N^{\frac{1}{\delta}})} \xi}^2 \Big)^{1/2}\\
    &\hspace{.1cm} \times    \Big(\int  dw \, \| \check{A}_{\ell_B}(x_1-x_2,x_1-x_3) \chi_{\ell_B}(x_1-w)\|_2\, \norm{(\cN_{\chi_{\ell_B}} (w) +1)^{\frac{(n-1)}{2}}\xi}^2\Big)^{1/2}\\
   & \leq C N^{\delta} \Big(\int  dw \, \norm{(\cN_{\chi_{\ell_B}} (w) +1)^{\frac{n}{2}}  \xi}^2 \Big)^{1/2}\ell_B^3\|   \check{A}_{\ell_B}\|_2   \Big(\int  dw \, \norm{(\cN_{\chi_{\ell_B}} (w) +1)^{\frac{(n-1)}{2}}\xi}^2\Big)^{1/2}\\
    & \leq C N^{\delta +\frac{3\kappa-2+7\epsilon}{2}} \la\xi, \cL_n\xi\ra.
\end{split} \]    
if $\delta$ is small enough. Here, we used (\ref{eq:localNboundexample}) to control $\cN_{2\ell_B}(w)$, and Lemma \ref{lemma:propertiesquadratickernels} to control $\ell_B^3\|   \check{A}_{\ell_B}\|^2_2\lesssim N^{3\kappa-2+7\epsilon}$. The residual term can be treated similarly
\begin{align*}
  &  X^{(1)}_\infty \leq C  \Big(\int  dw \, \norm{\cN(\cN_{\chi_{\ell_B}} (w) +1)^{\frac{n}{2}}  {\Theta(\cL_n/N^{\frac{1}{\delta}})} \xi}^2 \Big)^{1/2}\\
    &\hspace{.1cm} \times \| \check{A}_{\infty}\|_2   \Big(\int  dw \, \norm{(\cN_{\chi_{\ell_B}} (w) +1)^{\frac{(n-1)}{2}}\xi}^2\Big)^{1/2}\\
    & \leq \| \check{A}_{\infty}\|_2 N^{\frac{1}{n\delta }+3\kappa/2-6\epsilon} \la\xi, \cL_n\xi\ra\leq \la\xi, \cL_n\xi\ra,
\end{align*}
for $\delta$ small enough, where we have used \eqref{eq:compare_N_with_L} to control $\mathcal{N}$ and $\| \check{A}_{\infty}\|_2\leq C_m N^{-m}$ for any $m$ according to Lemma \ref{lemma:propertiesquadratickernels}. Choosing $0 < \delta < (2-3\kappa-7\epsilon)/2$, and proceeding for the terms $X^{(i)}_{\ell_B}+X^{(i)}_\infty$ with $i\in \{2,3\}$ in an analogue fashion, we conclude that 
 \[ \Big| \frac{d}{dt} \langle e^{t   {\Theta A^*}-\hc } \Omega, \cL_n     e^{t   {\Theta A^*}-\hc} \Omega\rangle \Big|  \leq C  \langle e^{t   {\Theta A^*}-\hc} \Omega, \cL_n     e^{t   {\Theta A^*}-\hc} \Omega\rangle  \]
 for every $t \in \bR$. Applying Gronwall's lemma in the interval $t \in [0;1]$, we obtain (\ref{eq:expectationL}). 

Finally, we prove (\ref{eq:AThetaA}). Since  $\mathcal{L}_n a_x^* a_y^* a_z^*=a_x^* a_y^* a_z^* \mathcal{L}_n(x,y,z)$ with 
\[\mathcal{L}_n(x,y,z):=  \int dw \, \big( \cN_{\chi_{\ell_B}} (w) +\chi_{\ell_B} (x) +\chi_{\ell_B} (y) +\chi_{\ell_B} (z)  + 1 \big)^n,\]
and since we have the trivial bound $  \mathcal{L}_n(x,y,z)\leq  4^n \mathcal{L}_n$, we find 
\begin{align*}
    Z: & =\langle e^{\Theta A^* - A\Theta} \Omega, A  \mathbbm{1}_{[1;\infty)} (\cL_n / N^{\frac{1}{\delta}}) \mathcal{N}^k A^* e^{\Theta A^* - A\Theta} \Omega \rangle \\
    & = \langle e^{\Theta A^* - A\Theta} \Omega, A  \mathbbm{1}_{[1;\infty)} (\cL_n / N^{\frac{1}{\delta}})\mathcal{N}^k A^* \mathbbm{1}_{[4^{-n},\infty)} (\cL_n / N^{\frac{1}{\delta}}) e^{ \Theta A^* - A\Theta} \Omega \rangle.
\end{align*}
With $\| (\cN+1)^{-(3+k)} A  \mathbbm{1}_{[1;\infty)} (\cL_n / N^{\frac{1}{\delta}})\mathcal{N}^k A^*  \| \leq C N^\tau$, for some $\tau > 0$, we conclude 
\[  Z \leq  C N^\tau \|  \mathbbm{1}_{[4^{-n},\infty)} (\cL_n / N^{\frac{1}{\delta}})  e^{\Theta A^* - A\Theta} \Omega \| \| (\cN+1)^{3+k} e^{\Theta A^* - A\Theta} \Omega  \|. \]
From \eqref{eq:norm1-Theta}, which clearly holds as well for $1-\Theta$ being replaced by $ \mathbbm{1}_{[4^{-n};\infty)} (\cL_n / N^{\frac{1}{\delta}})$, and \eqref{eq:easyboundNm} we obtain $Z  \leq C N^{-\nu}$ if $\delta > 0$ is chosen small enough. 
\end{proof}

\subsection{The number of particles in the trial state} 

 We recall the definition of the trial state
\begin{equation}\label{eq:psi1} \Psi_N = W_{N_0} e^{B-B^*} e^{\Theta A^*  - A \Theta } \Omega \end{equation} 
in (\ref{eq:psi0}), where $N_0$ represents the number of particles in the condensate created by the Weyl transformation $W_{N_0}$. Notably, the Bogoliubov transformation $e^{B-B^*}$ creates additional $\sum_{p \in \Lambda_+^*} \sigma_p^2$ many particles outside of the condensate, and hence we expect that $N_0 +\sum_{p \in \Lambda_+^*} \sigma_p^2$ approximates the total number of particles in our trial state $\Psi_N$. To guarantee that the total number of particles is close to $N$, we choose the number of particles in the condensate $N_0$ so that 
\begin{align}
    \label{eq:particle_number_0}
    \abs{N_0 +\sum_{p \in \Lambda_+^*} \sigma_p^2 - N}
\leq C N^{3\kappa/2-\epsilon}
\end{align}
with $\sigma_p = \sinh \mu_p$ as defined below \eqref{eq:definition_cut_off_Bog_objects}. 

In the next lemma, we check that, with these definitions, our trial state has the correct number of particles. To be precise, we find $N_-$ and $N_+$, both consistent with \eqref{eq:particle_number_0} such that the first corresponding state $\Psi_N^-$ has slightly less and the other state $\Psi_N^+$ slightly more than $N$ particles. As is explained in the proof of Theorem \ref{theorem:main}, the states $\Psi_N^\pm$ can be used to construct a state $\Gamma$ by a convex combination that has the exact right number of particles $N$. In the next section, we compute for any $N_0$ satisfying \eqref{eq:particle_number_0} the energy of the corresponding state $\Psi_N$. 
\begin{lemma}\label{lm:NN2}
The state $\Psi_N$ defined in (\ref{eq:psi1}) satisfies
\begin{align*}
   \abs{\la \Psi_N,\cN \Psi_N \ra - N_0 -\sum_{p \in \Lambda_+^*} \sigma_p^2}
\leq C N^{3\kappa/2-5\epsilon}.  
\end{align*}
In particular, we find $N_-,N_+$ such that $N_0:=N_\pm$ satisfies \eqref{eq:particle_number_0} and the corresponding states $\Psi^{\pm}_N$ satisfy
 \begin{align*}
     \la \Psi^{-}_N,\cN \Psi^{-}_N \ra+ N^{\frac{3\kappa}{2}-\epsilon} \leq N\leq \la \Psi^{+}_N, \cN \Psi^{+}_N\ra- N^{\frac{3\kappa}{2}-\epsilon}\leq 2 N.
 \end{align*}
 \end{lemma}

\begin{proof}
Using the explicit action of the Weyl operator $W_{N_0} $ and the Bogoliubov operator $ e^{B-B^*}$ on creation and annihilation operators $a_p^*$ and $a_p$, and the fact that taking the expectation with respect to the state $e^{A^* \Theta- \hc} \Omega$ of both $a_0, a_0^*$ and $a_p^* a_{-p}^*$ yields zero, we find
\begin{align*}
    \la \Psi_N, \cN \Psi_N \ra =
 N_0 + \sum_{p \in \Lambda^*_+} \sigma_p^2 
 + \la e^{A^* \Theta - \Theta A} \Omega, \sum_{p \in \Lambda_+^*} (\gamma_p^2 + \sigma_p^2) a_p^* a_p e^{A^* \Theta - \Theta A} \Omega\ra 
 . 
\end{align*}
Utilizing the fact that $\gamma_p^2 + \sigma_p^2\lesssim N^{\epsilon}$, together with \eqref{eq:easyboundNm} for $m=1$, yields
\begin{align*}
   \la e^{A^* \Theta - \Theta A} \Omega, \sum_{p \in \Lambda_+^*} (\gamma_p^2 + \sigma_p^2) a_p^* a_p e^{A^* \Theta - \Theta A} \Omega\ra 
\lesssim N^{\epsilon}\la e^{A^* \Theta - \Theta A} \Omega, \mathcal{N} e^{A^* \Theta - \Theta A} \Omega\ra \lesssim N^{3\kappa/2-5\epsilon}
 .
\end{align*}
\end{proof}

\section{Energy of the trial state}

Next, we compute the energy $\langle \Psi_N, \cH_N \Psi_N \rangle$ of the trial state $\Psi_N= W_{N_0} e^{B^* - B} \xi$ leading to the proof of Theorem \ref{theorem:Fockspaceresult}. 

\subsection{Action of Weyl and Bogoliubov transformations on $\cH_N$} 

We first extract the leading order energy in the next Lemma \ref{lemma:quadraticrenormalization} by applying the Weyl transformation $W_{N_0}$ and the Bogoliubov transformation $e^{B^*-B}$ on the Hamiltonian 
\begin{equation}\label{eq:HNFock}
    \cH_N = \cK + \cV_N
    = \sum_{p\in\Lambda^*} p^2 a_p^* a_p
    + \frac{N^\kappa}{2N} \sum_{p,q,r\in\Lambda^*} \hat{V}(r/N^{1-\kappa}) a_{p+r}^* a_{q-r}^* a_q a_p.
\end{equation} 
At this point, we extract the leading order and a contribution to the LHY-term and deal with the expectation of the remaining operators in the state $\xi$ later. The proof of the Lemma can be found in the appendix.
\begin{lemma}\label{lemma:quadraticrenormalization}
     {Let $N_0$ satisfy \eqref{eq:particle_number_0}.} Then, we have 
    \begin{equation}\label{eq:lmHN1}
    \begin{split}
        \la \Psi_N, \cH_N \Psi_N \ra 
        &= \; 4\pi\aa N^{1+\kappa}\\
        &+ \frac{1}{2}\sum_{p\in\Lambda^*_+}
        \Big(
        \sqrt{\abs{p}^4+2p^2 N^\kappa\widehat{Vf}(\frac{p}{N^{1-\kappa}})}
        - p^2 - N^\kappa\widehat{Vf}(\frac{p}{N^{1-\kappa}})
        + \frac{N^{2\kappa}\widehat{Vf}(\frac{p}{N^{1-\kappa}})^2}{2\abs{p}^2}
        \Big)
        \\
        &- \frac{N^\kappa \norm{\sigma}_2^2}{N} \sum_{p\in\Lambda^*}
        \hat{V}(p/N^{1-\kappa}) \eta_{\infty,p}
        - N^{\kappa-1} \sum_{p,r\in\Lambda^*}
        \hat{V}(r/N^{1-\kappa})
        \eta_{\infty,p+r}\sigma_{p}^2\\
        &+ \la \xi, (\cK + \cC_N^* + \cC_N + \cV_N)\xi \ra
        + \la \xi, (\cQ + \tilde\cC_N^* + \tilde\cC_N + \tilde\cV_N)\xi \ra + O(N^{5\kappa/2-\epsilon}+N^{2\kappa})
    \end{split}
    \end{equation}
    with
    \begin{equation}\label{eq:CtCtV}
    \begin{split} 
        \cC_N^* &=
        \frac{N^\kappa N_0^{1/2}}{N} \sum_{p, r\in\Lambda^*} \hat{V}(r/N^{1-\kappa}) 
        \sigma_p 
        a_{p+r}^* a_{-r}^* a_{-p}^*, \\
        \tilde\cC_N^* &=
        \frac{N^\kappa N_0^{1/2}}{N} \sum_{p, r\in\Lambda^*} \hat{V}(r/N^{1-\kappa}) 
        ((\gamma_{p+r}\gamma_r-1)\sigma_p + \sigma_{p+r}\sigma_r\gamma_p)
        a_{p+r}^* a_{-r}^* a_{-p}^*, \\
        \tilde\cV_N
        &=\frac{N^\kappa}{2N} \sum_{p,q,r\in\Lambda^*} \hat{V}(r/N^{1-\kappa}) 
        (\gamma_{p+r} \gamma_{q-r}\gamma_{q}\gamma_{p}-1
        +\sigma_{p+r}\sigma_{q-r}\sigma_{q}\sigma_{p}
        +2\gamma_{p+r}\sigma_{q-r}\sigma_{q}\gamma_p
        ) \\ &\hspace{5cm} \times 
        a_{p+r}^* a_{q-r}^*
         a_{q}a_{p}
        \\
        &\hspace{.4cm} +\frac{N^\kappa}{N} \sum_{p,q,r\in\Lambda^*} \hat{V}(r/N^{1-\kappa}) 
        \gamma_{p+r}\sigma_{q-r}\gamma_{q}\sigma_{p}
        a_{p+r}^*a_{-p}^*  a_{-q+r} a_{q} 
\end{split}
\end{equation}
and where $\cQ$ is a quadratic operator satisfying
$
    \pm\cQ \leq C N^{\kappa+\epsilon} \cN.
$
\end{lemma}

\subsection{Bounds for negligible cubic expectations} 
It is the content of this Subsection to demonstrate that the second term in the last line of \eqref{eq:lmHN1} does not contribute to the Lee-Huang-Yang order $N^{5\kappa/2}$. We will analyze the various terms $\cQ, \tilde\cC_N$ and $\tilde\cV_N$ separately in Lemma \ref{lm:boundN}, Lemma \ref{lm:tildeCN} and Lemma \ref{lm:tildeVN} respectively.

\begin{lemma}\label{lm:boundN}
There is a constant $C>0$ such that, for any $t\in \left[ 0;1\right]$, we have

\begin{equation*}\label{eq:N-bd}
        \la e^{t \Theta A^*-\hc}\Omega, \cN e^{t \Theta A^*-\hc}\Omega\ra
        \leq C N^{9\kappa/2 - 2 +\epsilon + 2\delta}.
    \end{equation*}
\end{lemma}

\begin{proof}
Denote $\xi_t = e^{t(\Theta A^* -A \Theta)} \Omega$ and observe that for any $m$ there is $C_m$ such that
\[ \begin{split} 
\big|& \la \xi_t, [ \cN , A ] \Theta \xi_t \ra \big| 
\\&\leq \frac{3}{\sqrt{N}} \int dx dy dz\, | \check{\eta} (x-y)| |\check{\sigma} (x-z)|  \| a_x a_y a_z \Theta \xi_t \| \| \xi_t \| \\ 
&\leq \frac{3}{\sqrt{N}} \norm{\eta}_2\norm{\sigma}_2
\left(\int dx \,   \| \cN_{\ell_B} (x)a_x \Theta \xi_t \|^2 \right)^{1/2} \norm{\xi_t}
+ C_m N^{-m} \big| \la \xi_t,  (\cN+1)^{3/2} \Theta \xi_t \ra \big|
\end{split} \]
where we split the integral according to the decomposition $\mathbbm{1}(|x-y|,|x-z|\leq \ell_B)+(1-\mathbbm{1}(|x-y|,|x-z|\leq \ell_B))$, and for the second term we used the rapid decay estimates in \eqref{eq:boundetam} and \eqref{eq:normseta}. Applying Lemma \ref{lemma:propertiesquadratickernels} together with the bounds \eqref{eq:localNbound}, \eqref{eq:easyboundNm}, and choosing $m$ large enough, we obtain 
\be 
\begin{split}
|\frac{d}{dt}\la \xi_t, \cN \xi_t\ra|\leq C \la \xi_t, \cN \xi_t\ra +CN^{9\kappa/2-2+\epsilon+2\delta}.
\end{split}
\ee
We conclude by Grönwall's Lemma.
\end{proof}

\begin{lemma}\label{lm:tildeCN}
Recall
\[ 
\tilde{\cC}_N =  N^{\kappa-1} N_0^{1/2} \sum_{p, r\in\Lambda^*} \hat{V}(r/N^{1-\kappa}) 
    ((\gamma_{p+r}\gamma_r-1)\sigma_p + \sigma_{p+r}\sigma_r\gamma_p)     a_{p+r} a_{-r} a_{-p}
    \]
    from \eqref{eq:lmHN1}.
For $\delta>0$ small enough, we have
    \begin{equation*}
      \big| \la e^{\Theta A^*-\hc}\Omega,  \tilde\cC_N e^{\Theta A^*-\hc}\Omega\ra \big| \leq CN^{5\kappa/2-3/4(2-3\kappa)+2\delta+\epsilon/2}.
    \end{equation*}
\end{lemma}
\begin{proof}
We write $\tilde\cC_N = \sum_{j=1}^4 \tilde\cC^{(j)}_N$, with, having switched to position space, 
\[ \tilde\cC^{(j)}_N =  \int dx dy dz \, K^{(j)} (x-y , x-z) a_x a_y a_z \]
for 
\[ \begin{split} K^{(1)} (y,z) &= N_0^{1/2} [V_N* (\widecheck{\gamma-1})](y) \check\sigma(z) \\
K^{(2)} (y,z) &= N_0^{1/2} \int dw \, [V_N * \check\gamma](w) (\widecheck{\gamma-1})(w+y)\check\sigma(w+z)  \\ 
K^{(3)} (y,z) &= N_0^{1/2}  [V_N * \check\sigma](y) \check\sigma (z)\\ 
K^{(4)} (y,z) &= N_0^{1/2} \int dw \, [V_N*  \check\sigma](w) \check\sigma (w+y)(\widecheck{\gamma-1})(w+z) \end{split} \]
where $V_N (x)$ denotes $N^{2-2\kappa} V (N^{1-\kappa}x)$.
The bounds from Lemma \ref{lemma:propertiesquadratickernels} imply $\| K^{(j)} \|_2 \leq C N^{(5\kappa -1)/2}$, for all $j=1, \dots ,4$. Moreover, for $|x|\geq \ell_\sigma$ and for any $m\geq 1$, 
\be
\begin{split}
|V_N\ast \check{\sigma}(x)|&\leq\int_{\substack{|y|\leq cN^{\kappa-1}}} |V_N(y)| |\check{\sigma}(x-y)|  dy \leq C_mN^{\kappa}  
\left(\frac{\ell_\sigma}{\abs{x}}\right)^m.
\end{split}
\ee
Here $c>0$ is a constant chosen so that supp$(V)\subset B_{c}(0)$, and we used \eqref{eq:boundsigmaxm} and $|x-y|\geq |x|/2$. The same estimate also holds if $\check{\sigma}$ is replaced by $\check{\gamma}$ or $(\widecheck{\gamma-1})$.  
It follows that, for $|x|\geq \ell_B$, and for any $m\geq 1$,
\be
\begin{split}
|K^{(4)}(x,y)|&\leq N^{1/2} |\int_{|w|\geq |x|/2} dw \, [V_N*  \check\sigma](w) \check\sigma (w+x)(\widecheck{\gamma-1})(w+y)|\\&+ N^{1/2} |\int_{|w|\leq |x|/2} dw \, [V_N*  \check\sigma](w) \check\sigma (w+x)(\widecheck{\gamma-1})(w+y)|\\&\leq  C_mN^{3/2}  
\left(\frac{\ell_\sigma}{\abs{x}}\right)^m \int dw \Big(N^{\kappa-1}|\check\sigma (w+x)|+|[V_N*  \check\sigma]|(w)\Big)|(\widecheck{\gamma-1})(w+y)|
\end{split}
\ee
A similar estimate holds when $|y|\geq \ell_B$ and for $|K^{(2)}(x,y)|$.
Hence, for $j=1,..,4$ and any $m\geq 1$  
\begin{equation}
\label{eq_bound_K_j}
\int_{\{|x|\geq \ell_B\}\cup \{|y|\geq \ell_B\}} dx dy|K^{(j)}(x,y)|^2 \leq C_mN^{-m}.
\end{equation}
To estimate the expectation of $\tilde\cC_N$, we first write
\begin{equation} 
\label{eq:dec_tilde_C_N}
\begin{split}
\langle \xi, \tilde\cC^{(j)}_N \xi \rangle=\langle \Theta \xi , \tilde\cC^{(j)}_N \Theta \xi \rangle+ \langle \xi, \tilde\cC^{(j)}_N (1-\Theta) \xi \rangle+ \langle (1-\Theta) \xi , \tilde\cC^{(j)}_N \Theta \xi \rangle.
\end{split}
\end{equation}
The two last terms can be controlled using Lemma \ref{lemma:propertiesL}, in particular \eqref{eq:norm1-Theta} and \eqref{eq:easyboundNm}, by 
\be
\begin{split}
|\langle (1-\Theta) \xi , \tilde\cC^{(j)}_N \Theta \xi \rangle|\leq \|(1-\Theta) \xi\| \|K^{(j)}\|_2\|(\cN+1)^{3/2}\xi\|\leq CN^{-1/(4\delta)}
\end{split}
\ee
if $\delta>0$ is chosen small enough.
For the main term on the r.h.s of \eqref{eq:dec_tilde_C_N}, we bound, for any $m\geq 1$, 
\be
\begin{split}
&|\langle \Theta \xi, \tilde\cC^{(j)}_N \Theta \xi \rangle | \\& \leq C \Big( \int_{\substack{|x-y|,\\|x-z|\leq \ell_B}} dx dy dz\|a_xa_ya_z\Theta \xi\|^2\Big)^{1/2}\|K^{(j)}\|_2\|\xi\|+C_mN^{-m}\langle \xi, (\cN+1)^{3/2}\xi \rangle,
\end{split}
\ee 
having
in particular used \eqref{eq_bound_K_j} to bound the contributions associated to $|x-y|\geq \ell_B$ and $|x-z|\geq \ell_B$.  Using the bound \eqref{eq:localNboundexample} and Lemma \ref{lm:boundN} for the first term, and the bound \eqref{eq:easyboundNm} for the second, we finally obtain, choosing $m$ sufficiently large, 
\be 
|\langle \Theta \xi, \tilde\cC^{(j)}_N \Theta \xi \rangle |\leq CN^{\delta} \|K^{(j)}\|_2\|\cN^{1/2}\xi\| \|\xi\|\leq  CN^{5\kappa/2-3/4(2-3\kappa)+2\delta+\epsilon/2}.
\ee
\end{proof}

Finally, we estimate the expectation $\la \xi , \tilde{\cV}_N \xi \ra$, with $\tilde{\cV}_N$ as defined in (\ref{eq:CtCtV}). 
\begin{lemma} \label{lm:tildeVN}
Let $\epsilon,\delta>0$ be small enough such that $\mu:=  3/4(2-3\kappa)-7\epsilon/2-3\delta/2>0$. Then,
\be
\begin{split}
\la e^{ \Theta A^*  - \hc} \Omega, \tilde{\cV}_N e^{  \Theta A^*  - \hc} \Omega \ra &\leq CN^{-\mu}\big(\la e^{ \Theta A^*  - \hc} \Omega, \cV_N e^{  \Theta A^*  - \hc} \Omega \ra +N^{5\kappa/2}\big).
\end{split}
\ee
\end{lemma} 

\begin{proof}
We will bound the expectation of 
\[ \tilde{\cV}_N^{(0)} = \frac{N^\kappa}{2N} \sum_{p,q,r\in\Lambda^*} \hat{V}(r/N^{1-\kappa}) 
        (\gamma_{p+r} \gamma_{q-r}\gamma_{q}\gamma_{p}-1) a_{p+r}^* a_{q-r}^* a_q a_p . \]
The rest of the contributions to $\tilde\cV_N$ can be bounded in a similar way, by changing factors of $\gamma-1$ for factors of $\sigma$. In  order to evaluate $\la \xi, \tilde{\cV}_N^{(0)} \xi \rangle$, we write $\theta_p = \gamma_p -1$ and $V_N (x) = N^{2-2\kappa} V (N^{1-\kappa}x)$, represent $\tilde{\cV}_N^{(0)}$ in position space to obtain
 \begin{equation}
\label{tildeV_N}
\begin{split}
\tilde\cV^{(0)}_N= \; &\frac{1}{2}\int dxdy \, V_N(x-y) \, a^* (\check{\theta}_x) a^* (\check{\theta}_y) a (\check{\theta}_x) a (\check{\theta}_y) \\&+ \int dxdy \, V_N(x-y) \, a^*_{x} a^* (\check{\theta}_y) a (\check{\theta}_x) a (\check{\theta}_y) + \text{h.c.}\\&+ \int dxdy \, V_N(x-y) \, a^*_{x}a^* (\check{\theta}_y) a_{x} a (\check{\theta}_y) \\&+ \int dxdy \, V_N(x-y) \, a^*_{x}a^* (\check{\theta}_y) a (\check{\theta}_x) a_{y} \\&+ \frac{1}{2}\int dxdy \, V_N(x-y) a^*_{x}a^*_{y}a (\check{\theta}_x) a (\check{\theta}_y) + \text{h.c.}\\&+\int dxdy \, V_N(x-y) \, a^*_{x}a^*_{y}a_{x}a (\check{\theta}_y) + \text{h.c.}
\end{split}
\end{equation}
where $\check{\theta}_x (s)$ denotes $ \check{\theta} (x-s)$, for any $x \in \Lambda $. To estimate the expectation of these terms, we insert the cutoff $\Theta$ on both sides of the inner product. For the first four terms on the r.h.s of \eqref{tildeV_N}, this is achieved in the same way as in the proof of Lemma \ref{lm:tildeCN}. For the last two terms, this is done differently and, as an example, we show it for the last one. We write
\begin{equation}\label{eq:insTheta} \begin{split} \int &dx dy \, V_N (x-y) \la \xi , a_x^* a_y^* a_x a (\check{\theta}_y) \xi \ra \\ =\; & 
 \int dx dy \, V_N (x-y) \la \Theta \xi , a_x^* a_y^* a_x a (\check{\theta}_y) \Theta \xi \ra +
  \int dx dy \, V_N (x-y) \la \xi , a_x^* a_y^* a_x a (\check{\theta}_y) (1-\Theta) \xi \ra \\ &+
   \int dx dy \, V_N (x-y) \la (1-\Theta) \xi , a_x^* a_y^* a_x a (\check{\theta}_y) \Theta \xi \ra. \end{split} \end{equation} 
The second (and similarly the third) term on the r.h.s. can be bounded by 
\[ \begin{split} \Big|   \int &dx dy \, V_N (x-y) \la \xi , a_x^* a_y^* a_x a (\check{\theta}_y) (1-\Theta) \xi \ra \Big| \\ \leq\; & \Big( \int dx dy \| a_x a_y (\cN+1) \xi \|^2 \Big)^{1/2} \Big( \int dx dy \, V_N^2 (x-y) \| a_x a (\check{\theta}_y) (\cN+1)^{-1} (1-\Theta) \xi \|^2 \Big)^{1/2} \\ \leq \; &\sup_y \| \check{\theta}_y \|_2 \| V_N \|_2 \| (\cN+1)^2 \xi \|  \| (1-\Theta) \xi \|. \end{split} \]
With Lemma \ref{lemma:propertiesL}, we obtain, for $\delta$ small enough,  
\[ \Big|   \int dx dy \, V_N (x-y) \la \xi , a_x^* a_y^* a_x a (\check{\theta}_y) (1-\Theta) \xi \ra \Big| \leq C N^{-1/(4\delta)} \]
Therefore, all that remains is to evaluate the first term on the r.h.s. of (\ref{eq:insTheta}).
It can be bounded, for any $m$, with
\be 
\begin{split}
&\Big|   \int dx dy \, V_N (x-y) \la \Theta \xi , a(\check{\theta}_x)^* a(\check{\theta}_y)^* a(\check{\theta}_x) a (\check{\theta}_y) \Theta \xi \ra \Big| \\&\leq \int _{|s|,|s'|,|t|,|t'|\leq \ell_B}ds dtds'dt' dx dy|V_N(x-y)||\check{\theta}(s)|^2 |\check{\theta}(t)|^2 \|a_{x-s'}a_{y-t'}\Theta\xi\|^2\\&\hspace{1cm}+CN^{-m}|\langle \xi, (\cN+1)^2\xi \rangle|\\&\leq CN^{\delta}\|\cN^{1/2}\Theta\xi\|^2\|\theta\|_2^4\|V_N\|_1\ell_B^3+CN^{-m+3\kappa-12\epsilon},
\end{split}
\ee
 having used \eqref{eq:localNbound}  for the contribution associated to $\mathbbm{1}(|s|,|t|,|s'|,|t'|\leq \ell_B)$, and \eqref{eq:boundsigmaxm} together with \eqref{eq:easyboundNm} to bound the remaining one. By Lemma \ref{lemma:propertiesquadratickernels}, Lemma \ref{lm:boundN} and $\|V_N\|_1\leq CN^{-1+\kappa}$, and choosing $m$ large enough we conclude that
\be 
\Big|   \int dx dy \, V_N (x-y) \la \Theta \xi , a(\check{\theta}_x)^* a(\check{\theta}_y)^* a(\check{\theta}_x) a (\check{\theta}_y) \Theta \xi \ra \Big| \leq CN^{5\kappa/2-3/2(2-3\kappa)+7\epsilon+3\delta}.
\ee
The second, third and fourth terms of \eqref{tildeV_N} can be estimated in a similar way. For the fifth one, we estimate, choosing $m$ large enough,
\be 
\begin{split}
&\Big|   \int dx dy \, V_N (x-y) \la \Theta \xi , a_x^* a_y^* a(\check{\theta}_x) a (\check{\theta}_y) \Theta \xi \ra \Big|\\& \leq \Big(\int _{|s|,|t| \leq \ell_B}ds dt dx dy|V_N(x-y)||\check{\theta}(s)|^2 |\check{\theta}(t)|^2 \|a_{x}a_{y}\Theta\xi\|^2\Big)^{1/2} \\& \hspace{2cm}\times \Big(\int _{|s|,|t|\leq \ell_B}ds dt dx dy|V_N(x-y)| \|a_{x-s}a_{y-t}\Theta\xi\|^2 \Big)^{1/2}\\&\hspace{1cm}+ C_mN^{-m}\langle \xi,(\cN+1)^{2}\xi \rangle   \\&\leq CN^{\delta/2} \|\cV_N^{1/2}\xi\|\|\cN^{1/2}\xi\|\|\theta\|_2^2\|V_N\|^{1/2}_1\ell_B^{3/2}\\&\leq CN^{-3/4(2-3\kappa)+7\epsilon/2+3\delta/2}(\langle \xi,\cV_N \xi \rangle +N^{5\kappa/2}).
\end{split}
\ee 
A similar bound holds for the last term of \eqref{tildeV_N}.

\end{proof}

\subsection{Bounds for relevant cubic expectations} 

In this subsection, we prove bounds on the expectations $\la \xi, (\cC_N^* + \cC_N) \xi \ra$, $\la \xi, \cV_N \xi \ra$ and $\la \xi, \cK \xi \ra$, for $\xi  = e^{\Theta A^*  -\hc} \Omega$. These conclude the energy estimate of the trial state $\Psi_N = W_{N_0} e^{B^* - B} \xi$ by inserting them in \eqref{eq:lmHN1}. We find it convenient to introduce the symmetrized version $\check{A}^{\mathrm{sym}}$ of the kernel $\check{A}(u,v):=N^{-\frac{1}{2}}\check{\eta}(u)\check{\sigma}(v)$ as
\begin{align*}
    \check{A}^{\mathrm{sym}}(u,v): = \! \check{A}(u,v) \! + \! \check{A}(v,u) \! + \! \check{A}(-  u,v \! - \! u) \! + \! \check{A}(v \! - \! u,- u)+\check{A}(-  v,u \! - \! v) \! + \! \check{A}(u \! - \! v,- v)  ,
\end{align*}
which is chosen exactly such that for all $(x_1,x_2,x_3)\in \Lambda^3$
\begin{align*}
   \check{A}^{\mathrm{sym}}(x_1-x_2,x_1-x_3)=\sum_{\sigma\in S_3} \check{A}(x_{\sigma(1)}-x_{\sigma(2)},x_{\sigma(1)}-x_{\sigma(3)}).
\end{align*}
In the following we denote $\check{A}^{\text{sym}}_{\ell_B}(u,v):= \mathbbm{1}(|u|,|v|\leq \ell_B)\check{A}^{\text{sym}}(u,v)$ and $\check{A}^{\text{sym}}_{\infty}(u,v):=\check{A}^{\text{sym}}(u,v)-\check{A}^{\text{sym}}_{\ell_B}(u,v)$. By Lemma \ref{lemma:propertiesquadratickernels}, for any $m\geq1$, we can find a constant $C_m>0$ such that $\|\check{A}_{\infty}\|_2 \leq C_mN^{-m}$.

\begin{lemma}
\label{lemma:C_N}
Recall the definition of the cubic operator from \eqref{eq:lmHN1}  
\begin{equation*}\label{eq:CN-pos} 
\cC_N=
         N_0^{1/2}\int dxdydz N^{2-2\kappa}V(N^{1-\kappa}(x-y))\check{\sigma}(x-z)a^*_xa^*_ya^*_z . 
       \end{equation*}
For $\delta>0$ small enough, we have
\begin{equation*}
\label{eq_Cn_f}
\begin{split} 
 \la e^{t\Theta A^*-\hc}\Omega, &(\cC_N+\cC_N^*) e^{t\Theta A^*-\hc}\Omega\ra \\ &= 2tN_0^{1/2}\int dxdy N^{2-2\kappa}V(N^{1-\kappa}x)\check{\sigma} (y) \check{A}^{\mathrm{sym}}(x,y) + \cE_{\cC}
 \end{split} 
 \end{equation*}
 where $\cE_{\cC}$ verifies
 \be 
 \pm \cE_{\cC} \leq CN^{-(2-3\kappa)+15\epsilon/2+2\delta}\Big(\int_0^t ds\; \la e^{ s\Theta A^*  - \hc} \Omega, \cV_N e^{  s\Theta A^*  - \hc} \Omega \ra+N^{5\kappa/2}\Big)
 \ee
 for all $t\in [0;1]$.
\end{lemma}

\begin{proof} 
Writing $\xi_t=e^{t\Theta A^*-\hc}\Omega$, we expand with Duhamel
\begin{equation*}
\label{eq:expansion_Duhamel}
\begin{split}
\la & \xi_t, (\cC_N+\cC_N^*) \xi_t\ra \\&= \int_0^t ds \; \la  \Theta \xi_s,[\cC_N,A^*]\Theta\xi_s \ra+ \la  \Theta \xi_s,[\cC_N,A^*](1-\Theta)\xi_s \ra\\&\hspace{1cm}+\langle \xi_s, \big( \cC_N (\Theta-1)A^*-(\Theta-1)\cC_N A^* -A\cC_N (\Theta-1)+ A(\Theta-1) \cC_N  \big) \xi_s \rangle + \hc
\end{split}
\end{equation*}
The terms containing a $(\Theta-1)$ factor can be bounded by $CN^{-1/(4\delta)}$, if $\delta$ is chosen small enough, by using Lemma \ref{lemma:propertiesL}, in particular \eqref{eq:norm1-Theta}, \eqref{eq:easyboundNm}, and \eqref{eq:AThetaA}, as well as the bounds $\|A^*(\cN+1)^{-3/2}\|\leq C\|A\|_2$ and $\|\cC_N(\cN+1)^{-3/2}\|\leq CN^{1/2}\|V_N\|_2\|\sigma\|_2$. Using again the notation $V_N(\cdot)=N^{2-2\kappa}V(N^{1-\kappa} \cdot)$, we observe that 
\be 
[\cC_N,A^*]= T_0+T_2+T_4
\ee
for
\be 
\begin{split}
T_0&=2N_0^{1/2}\int dxdy V_N(x)\check{\sigma} (y) \check{A}^{\mathrm{sym}}(x,y)\\ T_2&=N_0^{1/2}\int dx dy dz dz' \check{A}^{\mathrm{sym}}(x-y,x-z') \times  \\&\hspace{1cm}\Big[V_N(x-y)\check{\sigma}(x-z)+V_N(x-z)\big(\check{\sigma}(x-y)+ \check{\sigma}(z-y)\big)\Big]a^*_{z'}a_{z}\\
T_4&=\frac{N_0^{1/2}}{2}\int dx dy dz dx'dy'  \check{A}^{\mathrm{sym}}(x'-y',x'-z)\times  \\&\hspace{1cm}\Big[V_N(x-y)\check{\sigma}(x-z)+V_N(x-z)\big(\check{\sigma}(x-y)+ \check{\sigma}(z-y)\big)\Big]a^*_{x'}a^*_{y'}a_{x}a_{y}.
\end{split}
\ee 
We find, 
\be
\begin{split}
&|\la  \Theta \xi_s,T_4\Theta\xi_s \ra|\\&  \leq CN^{1/2} \|V_N\|_1^{1/2}\|\sigma\|_2\Big(\int_{\substack{|x'-y'|,\\|x'-z|\leq \ell_B}} dx'dy'dz\|a_{x'}a_{y'}\Theta\xi_s\|^2\Big)^{1/2}\\& \hspace{0.5cm}\times \|\check{A}^{\mathrm{sym}}\|_2\Big(\int_{\substack{|x-y|,\\|x-z|\leq \ell_B}}dxdydz(|V_N(x-y)|+|V_N(x-z)|)\|a_xa_y\Theta\xi_s\|^2\Big)^{1/2} \\& \leq CN^{1/2+\delta/2}\|V_N\|_1^{1/2}\ell_B^{3/2}\|\sigma\|_2\|\check{A}^{\mathrm{sym}}\|_2\|\cN^{1/2}\xi_s\|\big(N^{\delta/2}\|V_N\|_1^{1/2}\|\|\cN^{1/2}\xi_s\|+\ell_B^{3/2}\|\cV_N^{1/2}\xi_s\|\big)\\& \leq C N^{-(2-3\kappa)+7\epsilon+3\delta/2}(\langle \xi_s,\cV_N \xi_s \rangle+N^{5\kappa/2}) + CN^{5\kappa/2-7/4(2-3\kappa)+9\epsilon/2+3\delta}.
\end{split}
\ee
Here, we used the bound \eqref{eq:boundsigmaxm}, the estimate $\|\check{A}_{\infty}\|_2 \leq C_mN^{-m}$, $\text{supp}(V_N)\subset B_{cN^{\kappa-1}}(0)$, for some constant $c>0$ and \eqref{eq:easyboundNm} to show that the contribution associated to $1-\mathbbm{1}(|x-y|,|x-z|,|x'-y'|,|x'-z|\leq \ell_B)$ is negligible. For the main contribution, we applied Lemma \ref{lemma:propertiesquadratickernels}, Lemma \ref{lemma:propertiesL_pre} and Lemma \ref{lm:boundN}. Then for $T_2$, we find
\be
\begin{split}
|\la  \Theta \xi_s,T_2\Theta\xi_s \ra| &\leq CN^{1/2} \ell_B^3\|V_N\|_2\|\sigma\|_2 \|\check{A}^{\mathrm{sym}}\|_2 \la \xi_s, \cN\xi_s \ra \\&\leq CN^{5\kappa/2-(2-3\kappa)+15\epsilon/2+2\delta}
\end{split}
\ee
where we used again \eqref{eq:boundsigmaxm}, the estimate $\|\check{A}_{\infty}\|_2 \leq C_mN^{-m}$ and the support of $V_N$, to show that the contribution associated to $1-\mathbbm{1}(|x-y|,|x-z|,|x-z'|\leq \ell_B)$ can be neglected. For the main contribution, we applied Lemma \ref{lemma:propertiesquadratickernels} and Lemma \ref{lm:boundN}.

\end{proof}

Next, we give an expression for the expectation of the potential energy operator $\cV_N$. 
\begin{lemma} \label{lm:VN} 
Recall
\begin{equation*}
\cV_N= \frac{1}{2}\int dxdy N^{2-2\kappa}V(N^{1-\kappa}(x-y))a^*_xa^*_ya_x a_y
\end{equation*}
from \eqref{eq:lmHN1}.
We have, for $\delta, \epsilon>0$ small enough,
\begin{equation}
 \label{eq_Vn_f}
\begin{split} 
 \la e^{t\Theta A^*-\hc}\Omega, &\cV_N e^{t\Theta A^*-\hc}\Omega\ra \\ &= \frac{t}{2}\int dxdy  N^{2-2\kappa}V(N^{1-\kappa}x)\check{A}^{\mathrm{sym}}(x,y)^2 + O (N^{5\kappa/2-(2-3\kappa)+15\epsilon/2+2\delta})
\end{split} \end{equation}
for any $t\in [0;1]$.
\end{lemma}

\begin{proof}
Writing $\xi_s = e^{sA^* \Theta -\hc} \Omega$, we obtain, for any $t \in [0;1]$, 
    \begin{equation*}\label{eq:VN00} 
    \begin{split}
        &\la e^{ t\Theta A^*-\hc}\Omega, \cV_N e^{t \Theta A^*-\hc}\Omega\ra 
        = 
        \int_0^t ds \la \xi_s , \big( [\cV_N, A^*] + \hc \big) \xi_s \ra +\cE_\Theta
    \end{split}
    \end{equation*}
where $\cE_\Theta$ consists of all the contributions proportional to $(1-\Theta)$. We used, in particular, $[\cV_N, \Theta]=0$, since both $\cV_N$ and $\Theta$ are multiplication operators in position space for fixed number of particles. By Lemma \ref{lemma:propertiesL} and proceeding as in (\ref{eq:insTheta}), $\cE_\Theta$ can be included in the error by taking $\delta > 0$ small enough. We write 
 \begin{equation}
\label{eq:comVN}
\begin{split}
    [\cV_N,A^*] +\hc =&\; \frac{1}{2}\int dxdydz \, V_N (x-y) \check{A}^{\mathrm{sym}}(x-y,x-z) a_x^*a_y^*a_z^* + \hc \\ &+\frac{1}{2}\int dx dy dz_1 dz_2 \, V_N (x-y) \check{A}^{\mathrm{sym}}(z_1-z_2,z_1-x)a_x^*a_y^* a_{z_1}^* a_{z_2}^* a_y + \hc \\=& \; T_3+T_5, 
\end{split}
\end{equation}
where again $V_N (x)$ denotes $N^{2-2\kappa} V (N^{1-\kappa} x)$. To estimate $\la \xi_s, T_5 \xi_s \ra$, we apply Cauchy-Schwarz and Lemma \ref{lemma:propertiesL} to obtain
\begin{equation*}
    \begin{split}
\abs{\la \xi_s,T_{5} \xi_s \ra} 
&\leq \; C
\Big( \int dxdydz_1dz_2 
|V_N(x-y)| |\check{A}^{\mathrm{sym}}_{\ell_B}(z_1-z_2,z_1-x)|^2\|a_y \Theta\xi_s\|^2 \Big)^{1/2} \\
&\hspace{1cm}\times \Big( \int_{\substack{|z_1-z_2|,\\|z_1-x| \leq \ell_B}} dxdydz_1dz_2 |V_N(x-y)| \|a_{z_1}a_{z_2} a_xa_y\Theta\xi_s\|^2 \Big)^{1/2}
+\cE_\Theta \\ 
&\leq CN^{\delta} \|V_N\|_1^{1/2}  
 \|\check{A}^{\mathrm{sym}}\|_2 \|\cN^{1/2} \xi_s \| \| \cV_N^{1/2} \xi_s \| + \cE_\Theta  \end{split} 
\end{equation*} 
so that
\begin{equation*}\label{eq:T5-fin}
|\la \xi_s, T_5 \xi_s \ra |
\leq CN^{-5/4(2-3\kappa)+\epsilon+2\delta} (\la \xi_s , \cV_N \xi_s \ra+N^{5\kappa/2})
\end{equation*} 
for all $s \in [0;t]$. We now turn to the expectation of the cubic term $T_3$ in \eqref{eq:comVN}. Notice that this term has the same form as the cubic operator $\cC_N$ analyzed in Lemma \ref{lemma:C_N}, if the factor $N_0^{1/2} \check{\sigma} (x-z)$ is replaced by $\check{A}^{\mathrm{sym}} (x-y, x-z)/2$, and thus can be bounded similarly. Hence, 
\be
\langle \xi_s,T_3 \xi_s\rangle =  s\int dxdy  V_N(x)\check{A}^{\mathrm{sym}}(x,y)^2+ \cE_{T_3}
\ee
with 
\be 
\pm \cE_{T_3} \leq CN^{-(2-3\kappa)+15\epsilon/2+2\delta}\Big(\int_0^s dr\; \la \xi_r, \cV_N \xi_r \ra+N^{5\kappa/2}\Big).
\ee
All together, we therefore have 
\begin{equation}
\label{V_N_gr}
\la \xi_t, \cV_N \xi_t\ra = \frac{t}{2}\int dxdy  V_N(x)\check{A}^{\mathrm{sym}}(x,y)^2 +\cE_{\cV_N}
\end{equation}
with 
\be 
\pm \cE_{\cV_N} \leq CN^{-(2-3\kappa)+15\epsilon/2+2\delta}\Big(\int_0^t ds\; \la \xi_s, \cV_N \xi_s \ra+N^{5\kappa/2}\Big).
\ee
The first term on the r.h.s of \eqref{V_N_gr} is of order $CN^{5\kappa/2}$, as can be seen by the bounds in Lemma \ref{lemma:propertiesquadratickernels}. Taking $\epsilon$ and $\delta $ small enough, we therefore first get  $\la \xi_t, \cV_N \xi_t\ra \leq CN^{5\kappa/2}$ by applying Gronwall's Lemma. Inserting this bound in \eqref{V_N_gr} then implies \eqref{eq_Vn_f}.
\end{proof}

Lastly, we turn to the kinetic energy operator. 
\begin{lemma} \label{lm:Kxi}
Let $\mu:=\min\{\epsilon,1-3\kappa/2-\epsilon - 2 \delta,2-3\kappa-15\epsilon/2 -2\delta\}$. For $\epsilon,\delta>0$ small enough,
\[ \begin{split}  \la e^{A^* \Theta - \hc} &\Omega, \cK e^{A^* \Theta - \hc} \Omega \ra \\ &\leq  - N^{1/2} \int dx dy \, N^{2-2\kappa} (Vf) (N^{1-\kappa} x) \check{\sigma} (y) \check{A}^{\mathrm{sym}} (x,y) + O (N^{5\kappa/2-\mu/2}). \end{split} \]
\end{lemma} 

\begin{proof} 
Denote $\xi_t = e^{t  {\Theta A^*} - \hc} \Omega$ and compute
\begin{equation}\label{eq:Kexp} \begin{split} 
\la \xi_t, \cK \xi_t \ra & = \la \Omega, \cK \Omega \ra +  2 \text{Re } \int_0^t ds \, \la \xi_s , [ \cK, \Theta A^* ] \xi_t \ra \\ 
& = \; 2\text{Re } \int_0^t ds \, \la \xi_s, \Theta[\cK, A^*]  \xi_s \ra + 2\text{Re } \int_0^t ds \, \la \xi_s , [ \cK, \Theta ] A^* \xi_s \ra. \end{split} \end{equation} 
A straightforward computation exhibits with $\check{A}(u,v)=\frac{1}{\sqrt{N}}\check{\eta}(u)\check{\sigma}(v)$
\begin{align}
\nonumber 
  &  [ \cK, A^*]  = -2\int dx dy dz (\Delta_x+\Delta_y+ \Delta_z)\check{A}(x-y,x-z)a_x^* a_y^* a_z^*\\
  \nonumber 
    & =-2\int dx dy dz(\Delta_1 \check{A})(x-y,x-z)a_x^* a_y^* a_z^*-2\int dx dy dz \nabla_x (\nabla_2 \check{A})(x-y,x-z)a_x^* a_y^* a_z^*\\
    \label{Eq:Introduction_X_1_X_2}
    & =: X_1 + X_2.
\end{align}
We note at this point that by Lemma \ref{lemma:propertiesL} it follows that 
\begin{equation*} \label{eq:commKA2}
2 \text{Re } \int_0^t ds \, \la \xi_s, \Theta[ \cK, A^*]  \xi_s \ra = 2 \text{Re } \int_0^t ds \, \la \xi_s, [\cK, A^* ] \xi_s \ra + O (N^{-1/(4\delta)}) 
\end{equation*} 
if we choose the parameter $\delta >0$ in the definition of the cutoff $\Theta$ small enough.
{Splitting $\check{A}=\check{A}_{\ell_B}+\check{A}_\infty$ in the main contribution $\check{A}_{\ell_B}(u,v):=\mathbbm{1}(|u|,|v|\leq \ell_B)\check{A}(u,v)$ and the residuum $\check{A}_\infty$}, the $X_2$ term in (\ref{Eq:Introduction_X_1_X_2}) can be expressed and estimated as
\begin{align*}
  &  \text{Re } \la \xi_s, X_2\xi_s \ra = 4\text{Re }\int dx dy dz (\nabla_2 \check{A})(x-y,x-z)\la \xi_s,(\nabla a_x)^* a_y^* a_z^* \xi_s\ra \leq 4\|\cK^{1/2} \xi_s\|\\
  & \times \! \! \! \sum_{j\in \{\ell_B,\infty\} }\sqrt{\int dx dy dz dy' dz'(\nabla_2 \check{A}_j)(x-y,x-z)(\nabla_2 \check{A}_j)(x-y',x-z')\la \xi_s, a_{z'}a_{y'}a_y^* a_z^* \xi_s \ra }.
\end{align*}
For the main contribution $j=\ell_B$, we use the support properties of $\check{A}_{\ell_B}$ and Lemma \ref{lemma:propertiesL}
\begin{align*}
  &  \int dx dy dz dy' dz'(\nabla_2 \check{A}_{\ell_B})(x-y,x-z)(\nabla_2 \check{A}_{\ell_B})(x-y',x-z')\la \xi_s, a_{z'}a_{y'}a_y^* a_z^*\ra \\
  & \ \ \  \leq 4\left\|\nabla_2 \check{A}_{\ell_B}\right\|_2^2 \int dx \left\la \xi_s ,\left(\mathcal{N}_{\ell_B}(x)+2\right)^2\xi_s \right\ra 
  \lesssim N^{2\delta} \left\|\nabla_2 \check{A}_{\ell_B}\right\|^2=O\! \left(N^{4\kappa + 2\delta + \epsilon -1 }\right). 
\end{align*}
Using Lemma \ref{lemma:propertiesL} together with (\ref{eq:boundnablasigma}), one can show that the corresponding residual term $j=\infty$ is of the order $O(N^{-m})$ for any $m$ and hence  {
\begin{align*}
    \left|\la \xi_s, X_2\xi_s \ra\right|\leq CN^{2\kappa + \delta +\epsilon/2 -1/2 }\|\cK^{1/2} \xi_s\|.
\end{align*}
}
Concerning the $X_1$ term we note that the Fourier transform of $\Delta_1 \check{A}$ is given by $\frac{1}{2}\widetilde{\chi}(p)\left[N^{\kappa-1/2}\widehat{Vf}\! \left(\frac{p}{N^{1-\kappa}}\right)\sigma(q)\right]$ and we denote with $g$ the (inverse) Fourier transform of the smooth and compactly supported function $\widetilde{\chi}-1$. Then,
\begin{align*}
    \Delta_1 \check{A}(u,v) & =\frac{1}{2}N^{5/2-2\kappa}  \left(g*(Vf)(N^{1-\kappa}\cdot )\right)  (u)\check{\sigma}(v)+\frac{1}{2}N^{5/2-2\kappa} (Vf)(N^{1-\kappa}\cdot )  (u)\check{\sigma}(v)\\
    &= \frac{1}{2}\int_\Lambda  \mathrm{d}\zeta g(\zeta) \, N^{1/2}(Vf)_{N,\zeta}(u)\check{\sigma}(v)+ \frac{1}{2}N^{5/2-2\kappa} (Vf)_{N,0} (u)\check{\sigma}(v),
\end{align*}
where $(Vf)_{N,\zeta}(u):=N^{2-2\kappa}(Vf)(N^{1-\kappa}(u-\zeta))$. We observe that $N^{1/2}(Vf)_{N,\zeta}(u)\check{\sigma}(v)$ has almost exactly the same form as the kernel of the cubic operator $\cC_N$, just with the prefactor $N_0^{1/2}$ replacing by $N^{1/2}$ and with the potential $V$ replaced by $Vf$, and a shift $\zeta$. Denoting the modified cubic operator with the kernel $N^{1/2}(Vf)_{N,\zeta}(u)\check{\sigma}(v)$ by $\cC_{N,\zeta}$, we can proceed exactly as in Lemma \ref{lemma:C_N} and use Lemma \ref{lm:VN} with a shift $\zeta$ to obtain
\begin{align*}
    E(\zeta): & =\left\langle \xi_s,  (\cC_{N,\zeta} +\cC_{N,\zeta}^* )\xi_s\right\rangle-2s N^{1/2}\int dx dy \, (Vf)_{N,\zeta}(x) \check{\sigma} (y) \check{A}^{\mathrm{sym}} (x,y)\\
    & = O\! \left( N^{5\kappa/2 + 15\epsilon/2+2\delta-(2-3\kappa)}\right)
\end{align*}
uniformly in $\zeta\in \Lambda$. Using additionally that $\|g\|_{1}\lesssim 1$ yields
\begin{align}
\nonumber 
 &  \left\langle \xi_s,X_1 \xi_s \right\rangle  =-2s N^{1/2}\int_\Lambda \!  \mathrm{d}\zeta g(\zeta) \int dx dy \, (Vf)_{N,\zeta}(x) \check{\sigma} (y) \check{A}^{\mathrm{sym}} (x,y)\\
  \label{Eq:explicit_constant}
 &-2s N^{1/2}\int dx dy \, (Vf)_{N,0}(x) \check{\sigma} (y) \check{A}^{\mathrm{sym}} (x,y) + \! O\! \left( \! N^{5\kappa/2 + 15\epsilon/2+2\delta-(2-3\kappa)} \! \right) \! .
\end{align}
An explicit computation in Fourier space, using the bound \eqref{eq:sigmainftymomentum}, exhibits that the constant on the right hand side of \eqref{Eq:explicit_constant} is,  up to an error of the magnitude $O(N^{5\kappa/2-\epsilon})$, identical to $\lambda:= -2N^{1/2} \int dx dy \, N^{2-2\kappa} (Vf) (N^{1-\kappa} x) \check{\sigma} (y) \check{A}^{\mathrm{sym}} (x,y)$. We conclude for $\epsilon,\delta$ small enough, and $\mu$ as in the statement of the Lemma,
\begin{align}
\label{eq:main_theta_cK_I}
   2 \text{Re } \la \xi_s, \Theta [\cK, A^* ] \xi_s \ra \leq s\lambda + C N^{\frac{5\kappa}{4}-\frac{\mu}{2}}\|\cK^{1/2} \xi_s\|+CN^{\frac{5\kappa}{2} - \mu}.
\end{align}
Regarding the second term on the r.h.s. of (\ref{eq:Kexp}) we note that the cut-off operator $\Theta$ acts on the $m$-particle sector as multiplication by 
\begin{align*}
     F_m (x_1, \dots , x_m) : = \Theta \Big( \int_{\Lambda } dw \, \big( \sum_{j=1}^m \chi_{\ell_B} (x_j - w) + 1 \big)^n / N^{\frac{1}{\delta}} \Big),
\end{align*}
and similarly the action of $\mathcal{N}^{-2}[\cK, \Theta]$ is explicitly given by
\begin{equation*} 
\begin{split} - \frac{2}{m^2} \sum_{j=1}^m (\nabla_{x_j} F_m) (x_1, \dots , x_m) \cdot \nabla_{x_j} - \frac{1}{m^2}\sum_{j=1}^m (\Delta_{x_j} F_m) (x_1, \dots , x_m) .\end{split} \end{equation*} 
Using the bounds $|\nabla_{x_j} F_m (x_1, \dots , x_m)|^2  \leq C m \ell_B^{-2} $ and $|\Delta_{x_j} F_m (x_1, \dots , x_m)| \leq C m \ell_B^{-2}$, which follow directly from the definition of $F_m$, we obtain
\begin{align*}
   \| \frac{1}{\cN^2} [\cK , \Theta ] \xi_s \| \leq C \ell_B^{-2} \| (\cK+1)^{1/2} \xi_s\|.
\end{align*}
Since the function $\Theta'$ is supported in $[1,\infty)$, we have $ \frac{1}{\cN^2} [\cK , \Theta ]=  Q\frac{1}{\cN^2}[\cK , \Theta ]$, where $Q:=\mathbbm{1}_{[1,\infty)}(\cL_n / N^{\frac{1}{\delta}})$, and together with Lemma \ref{lemma:propertiesL} we obtain for any $\nu$ 
\begin{align*}
\Big|   \la \xi_s , [ \cK, \Theta ] A^*  \xi_s \ra \Big| \leq C \ell_B^{-2}  \| \cN^2 Q A^* \xi_s \| \| (\cK + 1)^{1/2} \xi_s \| \leq C_\nu N^{-\nu} \| (\cK+1)^{1/2} \xi_s \|,
\end{align*}
 if $\delta $ is small enough. In combination with (\ref{eq:Kexp}) and (\ref{eq:main_theta_cK_I}) we obtain
\begin{align*}
   \| (\cK+1)^{1/2} \xi_t \|^2= \la \xi_t, \cK \xi_t \ra\leq \lambda/2 + CN^{\frac{5\kappa}{2} - \mu} + CN^{\frac{5\kappa}{4}-\mu /2}\int_0^t ds \| (\cK+1)^{1/2} \xi_s \| .
\end{align*}
Grönwall's inequality therefore tells us that $\la \xi_1, \cK \xi_1 \ra\leq  e^{CN^{-\mu/2}} \! \left( \lambda/2+CN^{5\kappa/2-\mu/2}\right)$, which concludes the proof together with $ \lambda\lesssim N^{5\kappa/2}$.
\end{proof}

\subsection{Proof of Theorem \ref{theorem:Fockspaceresult}} 

We combine the results of Lemma \ref{lemma:C_N}, Lemma \ref{lm:VN} and Lemma \ref{lm:Kxi}.
\begin{lemma} \label{lm:cubic-exp} 
Recall $\cC_N$ from \eqref{eq:CtCtV} and $\cV_N$ as well as $\cK$ from \eqref{eq:HNFock}. There is $h>0$ such that for sufficiently small $\delta,\epsilon >0$, we have 
\[ \begin{split} 
\la \xi, (\cK + &\cC_N + \cC_N^* + \cV_N) \xi \ra = \; N^{\kappa-1} \sum_{p,q \in \Lambda^*} \hat{V} (p/N^{1-\kappa}) (\eta_{\infty,p} + \eta_{\infty,p+q}) \sigma_q^2  
 +O(N^{5\kappa/2-h}) .
        \end{split}\]
\end{lemma} 

\begin{proof}
First note that all the error terms in the aforementioned lemmata are of order $O(N^{5\kappa/2-h})$ if $\delta,\epsilon >0$ are small enough. We combine the lemmata \ref{lemma:C_N} and \ref{lm:Kxi} to obtain
\begin{equation*} \begin{split} \langle \xi, (\cK + \cC_N + \cC_N^*) \xi \ra &= \; N^{\kappa-1} \sum_{p,q \in \Lambda^*} \big( 2 \widehat{Vw}  (p/N^{1-\kappa}) + \widehat{Vf} (p/N^{1-\kappa}) \big) \sigma_q \\ &\times (\eta_{p}\sigma_{q}
	+\eta_{q}\sigma_{p}
	+\eta_{p+q}\sigma_{q}
	+\eta_{q}\sigma_{p+q}
	+\eta_{p}\sigma_{p+q}
	+\eta_{p+q}\sigma_{p}
	)
+O(N^{5\kappa/2-h})\\
&= \; N^{\kappa-1} \sum_{p,q \in \Lambda^*} \big( 2 \widehat{Vw}  (p/N^{1-\kappa}) + \widehat{Vf} (p/N^{1-\kappa}) \big) \sigma_q^2 (\eta_{p}
	+\eta_{p+q}
	)
+O(N^{5\kappa/2-h})
\end{split} \end{equation*} 
after using first \eqref{eq:particle_number_0} to replace $N_0$ by $N$ in the term coming from $\cC_N$. 
Note regarding the second equality that all other terms are bounded by $N^{\kappa-1} \norm{\eta}_2 \norm{\sigma}_2 \norm{\sigma}_1 \leq C N^{5\kappa/2 +\epsilon/2 - (2-3\kappa)/4} $ thanks to Lemma \ref{lemma:propertiesquadratickernels}.
For $\cV_N$ on the other hand we use Lemma \ref{lm:VN} to write
\[
\begin{split}
\la \xi, \cV_N \xi \ra 
&= \frac{N^{\kappa-2}}{2} \sum_{p,q,r \in \Lambda^*} \hat{V} (r/N^{1-\kappa}) (\eta_{p}\sigma_{q}
	+\eta_{q}\sigma_{p}
	+\eta_{p+q}\sigma_{q}
	+\eta_{q}\sigma_{p+q}
	+\eta_{p}\sigma_{p+q}
	+\eta_{p+q}\sigma_{p}
	)\\
	&\times
	(\eta_{p-r}\sigma_{q}
	+\eta_{q}\sigma_{p-r}
	+\eta_{p-r+q}\sigma_{q}
	+\eta_{q}\sigma_{p-r+q}
	+\eta_{p-r}\sigma_{p-r+q}
	+\eta_{p-r+q}\sigma_{p-r}
	) 
 + O(N^{5\kappa/2-h}) \\
 &= \frac{N^{\kappa-2}}{2} \sum_{p,q,r \in \Lambda^*} \hat{V} (r/N^{1-\kappa}) (\eta_{p}
	+\eta_{p+q}
	)
	(\eta_{p-r}
	+\eta_{p-r+q}
	) \sigma_{q}^2
 + O(N^{5\kappa/2-h})\\
 &= N^{\kappa-2} \sum_{p,q,r \in \Lambda^*} \hat{V} (r/N^{1-\kappa}) (\eta_{p}
	+\eta_{p+q}
	)
	\eta_{\infty,p-r}
	\sigma_{q}^2
 + O(N^{5\kappa/2-h})\\
 &= -N^{\kappa-1} \sum_{p,q \in \Lambda^*} \widehat{Vw} (r/N^{1-\kappa}) (\eta_{p}
	+\eta_{p+q}
	)
	\sigma_{q}^2
 + O(N^{5\kappa/2-h})
 .
 \end{split}\]
 For the second equality we used Lemma \ref{lemma:propertiesquadratickernels} to bound all the other terms either as 
 $N^{\kappa-2} \norm{\eta}_2 \norm{\sigma}_2 \norm{\eta}_1 \norm{\sigma}_1 \leq C N^{5\kappa/2 +\epsilon/2 - (2-3\kappa)/4} $ or as $N^{\kappa-2} \norm{\eta}_2^2 \norm{\sigma}_1^2 \leq C N^{5\kappa/2 +\epsilon - (2-3\kappa)/2} $. On the other hand, for the third equality we first simplified the expression exploiting the symmetry of the kernels and then removed the cutoff in $\eta$ yielding a small error due to \eqref{eq:eta-etainftysum} as $N^{\kappa-2} \norm{\eta-\eta_\infty}_1 \norm{\eta}_1 \norm{\sigma}_2^2 \leq C N^{5\kappa/2 -\epsilon}$. Furthermore, for the last equality we plugged in the scattering equation \eqref{eq:scatteringdiscrete}; observe that to absorb the finite size error of order $O(N^{2\kappa})$ we require in particular $h\leq \kappa/2$. To conclude the proof of the Lemma, note that after combining the above terms the claim follows by replacing $\eta$ by $\eta_\infty$ similarly to before.
\end{proof} 

We are now ready to show Theorem \ref{theorem:Fockspaceresult}.
\begin{proof}[Proof of Theorem \ref{theorem:Fockspaceresult}]
Let $\Psi_N^{\pm}$ be as in Lemma \ref{lm:NN2} and pick $\epsilon, \delta$ are small enough. Combining the result on the quadratic conjugation in Lemma \ref{lemma:quadraticrenormalization} with Lemma \ref{lm:cubic-exp} for the main terms therein and Lemma \ref{lm:boundN}, \ref{lm:tildeCN} and \ref{lm:tildeVN} for the negligible ones we have for some $h>0$
\begin{equation*} \label{eq:S} \begin{split} 
\la \Psi_N^{\pm}, &\cH_N \Psi_N^{\pm} \ra =  4\pi\aa N^{1+\kappa}
        + S
+ O(N^{5\kappa/2-h}) .
 \end{split} \end{equation*}
 where $S=\frac{1}{2}\sum_{p\in\Lambda^*_+}
        \Big(
        \sqrt{\abs{p}^4+2p^2 N^\kappa\widehat{Vf}(\frac{p}{N^{1-\kappa}})}
        - p^2 - N^\kappa\widehat{Vf}(\frac{p}{N^{1-\kappa}})
        + \frac{N^{2\kappa}\widehat{Vf}(\frac{p}{N^{1-\kappa}})^2}{2\abs{p}^2}
        \Big).
        $
Observe that the summand decays as $N^{3\kappa} |p|^{-4}$ for large $|p|$, allowing to restrict the sum to $|p| \leq N^{\kappa/2+h}$ up to the desired precision. We introduce 
\[ g_p (t) = 
        \sqrt{p^4+2p^2 N^\kappa\widehat{Vf}(\frac{tp}{N^{1-\kappa}})}
        - p^2 - N^\kappa\widehat{Vf}(\frac{tp}{N^{1-\kappa}})
        + \frac{N^{2\kappa}\widehat{Vf}(\frac{tp}{N^{1-\kappa}})^2}{2\abs{p}^2}
        \]
        in order to connect $S$ to the corresponding sum involving $\aa$. We find
 \[ 
 \begin{split} 
 S =& \; \frac{1}{2}\sum_{\substack{p\in\Lambda^*_+\\ \abs{p}\leq N^{\kappa/2+h}}}
         \left(
        \sqrt{|p|^4+16\pi \mathfrak{a} N^\kappa p^2}
        - p^2 -  8\pi \mathfrak{a}N^\kappa 
        + \frac{N^{2\kappa}(8\pi \mathfrak{a})^2}{2\abs{p}^2}
        \right)\\
        &+ \frac{1}{2}\int_0^1 ds\int_0^s dt \sum_{\substack{p\in\Lambda^*_+\\ \abs{p}\leq N^{\kappa/2+h}}}
        \frac{d^2}{dt^2} g_p(t)
        +O(N^{5\kappa/2-h})
         \end{split} 
         \]
using $\widehat{Vf} (0) = 8\pi \aa$ as well as $g'_p (0) = 0$, due to symmetry. It is easy to verify that for all relevant $p$ we have
$
            \abs{g_p''(t)} 
            \leq C N^{4\kappa-2} \frac{N^{\kappa}}{\abs{p}^2}
$
implying
\[         \begin{split}
        S
        &= \frac{1}{2}\sum_{\substack{p\in\Lambda^*_+}}
        \left(
        \sqrt{p^4+16\pi \mathfrak{a} N^\kappa p^2}
        - p^2 -  8\pi \mathfrak{a}N^\kappa
        + \frac{N^{2\kappa}(8\pi \mathfrak{a})^2}{2\abs{p}^2}
        \right)
        +O(N^{5\kappa/2-h}).\end{split} \]
         Note that we removed the momentum restriction again.
Replacing the Riemann sum by the usual Lee-Huang-Yang integral and evaluating it concludes the proof.
\end{proof}

\section*{Appendix}
It is the content of the Appendix to verify the auxiliary Lemma \ref{lemma:scattering} and Lemma \ref{lemma:propertiesquadratickernels} concerning the scattering coefficients, and Lemma \ref{lemma:quadraticrenormalization} which concerns the quadratic renormalization.

\begin{proof}[Proof of Lemma \ref{lemma:scattering}]
In the following let $\chi$ be a smooth function with support in $(-1/2,1/2)^3$ and $\chi(x)=1$ for $|x|\leq \delta$ and a suitable $\delta>0$. Using the fact that $w:=1-f$ is identical to $\frac{\mathfrak{a}}{|x|}$ for $x$ outside of the support of $V$, and choosing $N$ large enough such that $\chi=1$ on the support of $V(N^{1-\kappa}\cdot )$, we obtain
\begin{align*}
    -\Delta_x \left(N w(N^{1-\kappa}x)\chi(x)\right)=N^{3-2\kappa}Vf/2(N^{1-\kappa}x)+N^\kappa g(x),
\end{align*}
where $g$ is a smooth function with support in $(-1/2,1/2)^3$. Denoting with $\mathcal{F}$ the Fourier transform on the torus $\Lambda$, we obtain according to the definition of $\eta_{\infty,p}$,
\begin{align*}
    \eta_{\infty,q}=-\cF\! \left[N w(N^{1-\kappa}\cdot )\chi\right](q)+N^{\kappa}\frac{\cF[g](q)}{q^2}.
\end{align*}
Together with $ V(N^{1-\kappa}\cdot )\chi=V(N^{1-\kappa}\cdot)$ and analogously $ (Vw)(N^{1-\kappa}\cdot )\chi=(Vw)(N^{1-\kappa}\cdot)$ as well as $\hat{V}(\frac{p}{N^{1-\kappa}})=N^{3-3\kappa}\cF\! \left[V(N^{1-\kappa}\cdot )\right](p)$, we observe
\begin{align*}
    \frac{N^\kappa}{2N} \sum_{q\in \Lambda_+^*} \hat{V}(\frac{p-q}{N^{1-\kappa}})\eta_{\infty,q}
    & =-\frac{N^{3-2\kappa}}{2}\cF\! \left[(Vw)(N^{1-\kappa}\cdot )\chi\right] (p) +\frac{N^{2\kappa-1}}{2} \sum_{q\in \Lambda_+^*} \hat{V}(\frac{p-q}{N^{1-\kappa}})\frac{\cF[g](q)}{q^2}\\
    &=  -\frac{N^{3-2\kappa}}{2}\cF\! \left[(Vw)(N^{1-\kappa}\cdot )\right](p)+O(N^{2\kappa-1}),
\end{align*}
which concludes the proof of \eqref{eq:scatteringdiscrete} as
\begin{align*}
    \frac{N^{3-2\kappa}}{2}\cF\! \left[(Vw)(N^{1-\kappa}\cdot )\right](p)=\frac{N^{\kappa}}{2}\widehat{Vw}(p/N^{1-\kappa})=p^2 \eta_{\infty,p}+\frac{N^{\kappa}}{2}\hat{V}(\frac{p}{N^{1-\kappa}}).
\end{align*}
and the proof of (\ref{eq:etacorrection}) as
\begin{align*}
    \widehat{Vw} (0) = \hat{V} (0) - \widehat{Vf} (0) = \hat{V} (0) - 8\pi \frak{a}.
\end{align*}
The bounds \eqref{eq:sigmainftymomentum} to \eqref{eq:gammasigmainfty-etainftysum} are easy to check using 
\begin{equation}
\begin{split}\label{eq:sigmainfty}
&\sigma_{\infty,p}^2 = \frac{p^2 +  N^\kappa \widehat{Vf}(p/N^{1-\kappa})-\sqrt{p^4 + 2p^2 N^\kappa \widehat{Vf}(p/N^{1-\kappa})}}{2\sqrt{p^4 + 2p^2 N^\kappa \widehat{Vf}(p/N^{1-\kappa})}}
=\frac{1-2\eta_{\infty,p}-\sqrt{1-4\eta_{\infty,p}}}{2\sqrt{1-4\eta_{\infty,p}}},\\
&\gamma_{\infty,p}^2 = \frac{p^2 +  N^\kappa \widehat{Vf}(p/N^{1-\kappa})+\sqrt{p^4 + 2p^2 N^\kappa \widehat{Vf}(p/N^{1-\kappa})}}{2\sqrt{p^4 + 2p^2 N^\kappa \widehat{Vf}(p/N^{1-\kappa})}}
=\frac{1-2\eta_{\infty,p}+\sqrt{1-4\eta_{\infty,p}}}{2\sqrt{1-4\eta_{\infty,p}}},\\
&\gamma_{\infty,p}\sigma_{\infty,p} = \frac{-N^\kappa \widehat{Vf}(p/N^{1-\kappa})}{2\sqrt{p^4 + 2p^2 N^\kappa \widehat{Vf}(p/N^{1-\kappa})}}
=\frac{\eta_{\infty,p}}{\sqrt{1-4\eta_{\infty,p}}}.
\end{split}
\end{equation} 
\end{proof}
Before proving Lemma \ref{lemma:scattering}, we define for $m \in \bN$
 \begin{equation}\label{eq:defFm}
    F_{m}(\xi) =  \Big( 1+ \sum_{k \in \bN^3: |k| \leq m} \abs{\widehat{ x^k Vf} (\xi) }]\Big)^m - 1.
\end{equation}
and observe the elementary bounds
\begin{align}
    \label{eq:inftynormFm}
    \norm{F_m}_\infty
    &\leq C_m,\\
    \label{eq:normFm}
    \sum_{p\in\Lambda^*} F_{m}^2((p+\zeta)/N^{1-\kappa})
    &\leq C_m N^{3-3\kappa}
\end{align}
for all $N \in \bN$ and all $\zeta\in\mathbb{R}^3$. With \eqref{eq:inftynormFm} and \eqref{eq:normFm} at hand, we are in a position to verify Lemma \ref{lemma:propertiesquadratickernels}.

\begin{proof}[Proof of Lemma \ref{lemma:propertiesquadratickernels}]
The bounds in momentum space follow directly from the definitions, compare to Lemma \ref{lemma:scattering} for the hyperbolic functions.
For the decay estimates in position space we use smoothness in momentum space.
Note
\begin{equation*}
\partial_{p_j}\mu_\infty(p) 
= 
-\frac{1}{4} \left( 1+ \frac{2N^\kappa \widehat{Vf}(\frac{p}{N^{1-\kappa}})}{p^2} \right)^{-1} 
\partial_{p_j}\left( \frac{2N^\kappa \widehat{Vf}(\frac{p}{N^{1-\kappa}})}{p^2} \right) \, .
\end{equation*}
Thus, for $m\geq 1$
\begin{equation*}
\begin{split}
\abs{\mathbbm{1}_{[\ell_\sigma^{-1},\infty)}(\abs{p})\partial_{p_j}^m\mu_\infty(p)}
&\leq C_m \mathbbm{1}_{[\ell_\sigma^{-1},\infty)}(\abs{p})
\sum_{k=1}^m
\sum_{\substack{m_1, \dots ,m_k \geq 1 : \\  \sum_{i=1}^k m_i = m}} 
\prod_{i=1}^{k}
\Big| \frac{p^2}{p^2+N^\kappa} \partial_{p_j}^{m_i}\left( \frac{N^\kappa \widehat{Vf}(\frac{p}{N^{1-\kappa}})}{p^2} \right) \Big| \\
&\leq
C_m \ell_\sigma^m 
\mathbbm{1}_{[\ell_\sigma^{-1},\infty)}(\abs{p})
F_{m}(p/N^{1-\kappa}) 
\frac{N^\kappa}{\abs{p}^2}
\end{split} \end{equation*}
where we used \eqref{eq:defFm} and that $N^{-(1-\kappa)}\leq \ell_\sigma$.
Recalling that $\chi (p) = \chi_l (\abs{p}\ell_\sigma)$, we find $
    \abs{\partial_{p_j}^m \chi (p)} 
    \leq C_m \ell_\sigma^m
    \mathbbm{1}_{[\ell_\sigma^{-1}, \infty)}(\abs{p})
$.
As a result, for $m \geq 1$,  
\begin{equation*}
    \abs{\partial_{p_j}^m\mu(p)} 
    \leq C_m  \ell_\sigma^m \log N
\mathbbm{1}_{[\ell_\sigma^{-1},\infty)}(\abs{p})
F_{m}(p/N^{1-\kappa}) 
\frac{N^\kappa}{\abs{p}^2}
    \end{equation*}
where the factor $\log N$ arises from the term in which all derivatives act on $\chi$. 
Using this estimate, we can analyze $\sigma (p)$ and obtain for $m \geq 1$
\[
\begin{split}
    \abs{\partial_{p_j}^m\sigma(p)}
    \leq \; &
    C_m
    \sum_{k=1}^m \abs{\sinh^{(k)}(\mu (p))}
     \sum_{\substack{m_1,\dots,m_k\geq 1: \\ \sum_{i=1}^k  m_i =m}} \prod_{i=1}^k
    \abs{\partial_{p_j}^{m_i} \mu(p)
    }\\
    \leq \; & 
    C_m
    (\norm{\sigma}_\infty + \norm{\gamma}_\infty)    (\log N)^m  \ell_\sigma^m
\mathbbm{1}_{[\ell_\sigma^{-1},\infty)}(\abs{p})
F_{m}(p/N^{1-\kappa}) 
\frac{N^\kappa}{\abs{p}^2}\\
\leq \; & 
    C_m
    N^\epsilon  \ell_\sigma^m
\mathbbm{1}_{[\ell_\sigma^{-1},\infty)}(\abs{p})
F_{m}(p/N^{1-\kappa}) 
\frac{N^\kappa}{\abs{p}^2}
\end{split} \]  
We now prove the decay estimates in position space.
For any $x\in\Lambda  = [-1/2 ; 1/2]^3$ and any integer $m\geq 1$, we have
\begin{equation*}
\begin{split} 
c^m &\abs{x_j}^m \abs{\check\sigma (x) }
    \leq
    \abs{(e^{2\pi i x_j}-1)^m\check\sigma(x) }
    = \abs{(e^{2\pi i x_j}-1)^m \sum_{p\in\Lambda^*} \sigma(p)e^{ip\cdot x}}\\
    &= \abs{(e^{2\pi i x_j}-1)^{m-1} \sum_{p\in\Lambda^*} (\sigma (p-2\pi e_j)-\sigma (p))e^{ip\cdot x}}\\
    &= \Big| \int_0^{2\pi} ds \sum_{p\in\Lambda^*}  (e^{2\pi i x_j}-1)^{m-1} \partial_{p_j}\sigma (p- s e_j)e^{ip\cdot x} \Big| \\
&\leq  \int_0^{2\pi} ds_1\dots \int_0^{2\pi} ds_m \sum_{p\in\Lambda^*}  \Big| \partial_{p_j}^m\sigma(p-\sum_{i=1}^m s_i e_j) \Big| 
\leq C \int_0^{2\pi m} ds \sum_{p\in\Lambda^*}  \abs{ \partial_{p_j}^m\sigma (p-s e_j)}\\
&\leq \; C_m
    N^\epsilon  \ell_\sigma^m
    (
\norm{F_{m}}_\infty
\norm{N^\kappa/\abs{\cdot}^2 \mathbbm{1}(\abs{\cdot}\leq N^{1-\kappa})}_1
+
\norm{F_{m}}_2
\norm{N^\kappa/\abs{\cdot}^2 \mathbbm{1}(\abs{\cdot}\geq N^{1-\kappa})}_2
)\\
&\leq \; C_m
    N^{1+\epsilon}  \ell_\sigma^m
\end{split}
\end{equation*}
Bounds for $\gamma-1$ can be proven analogously; this implies (\ref{eq:boundsigmaxm}).
Also the bound \eqref{eq:boundetam} for $\check\eta$ can be proven in the same manner noting that each derivative yields either a factor $\ell_\eta$ (for the ones hitting the cutoff) or $N^{-(1-\kappa)},\abs{p}^{-1}\leq \ell_\eta$ (for the ones hitting $\eta_\infty$), and the same goes for the proof of \eqref{eq:boundnablasigma}.

The first part of equation \eqref{eq:boundsigmax1} follows by using Cauchy-Schwarz on a ball of size $\ell_\sigma N^{\epsilon/3}$ and the rapid decay beyond whereas the second part is just $\norm{\check{\sigma}}_\infty \leq \norm{\sigma}_1$

Lastly, we sketch how to adapt these argument for showing \eqref{eq:boundnablasigma}. We have
\begin{align*}
\begin{split}
    \int_{|x|\geq r}dx | \nabla \check  \sigma(x)|^2
    &\leq C_m \sum_{i=1}^3\int_{|x|\geq r}dx | \nabla \check  \sigma(x)|^2 \left(\frac{\abs{x_i}}{r}\right)^{2m}\\
    &\leq C_m r^{-2m}\sum_{i=1}^3\int dx \abs{(e^{2\pi i x_i}-1)^{m}\sum_{p\in\Lambda^*} p\sigma_p e^{ipx}}^2\\
   &\leq C_m r^{-2m} \sum_{i=1}^3  
   \sum_{p\in\Lambda^*} \left \vert   \int_0^{2\pi} ds_1\dots \int_0^{2\pi} ds_m \partial_{p_i}^m (p\sigma(p))[p-\sum_{k=1}^m s_k e_i] \right \vert^2\\
	&\leq C_m r^{-2m} \sum_{i=1}^3  
   \int_0^{2\pi m} ds \sum_{p\in\Lambda^*} \left \vert  \partial_{p_i}^m (p\sigma(p))[p-s e_i] \right \vert^2\\
   &\leq C_m N^{2\epsilon} \sum_{i=1}^3 \int_0^{2\pi m} ds \sum_{p\in\Lambda^*\colon\abs{p}\geq \ell_\sigma^{-1}}
   \abs{F_{m}((p-se_i)/N^{1-\kappa})}^2
\frac{N^{2\kappa}}{\abs{p}^2} \left(\frac{\ell_\sigma}{r}\right)^{2m}\\
   &\leq C_m N^{1+\kappa + 2\epsilon}
\left(\frac{\ell_\sigma}{r}\right)^{2m}.
\end{split}
\end{align*}
\end{proof}

Finally, we provide a proof of Lemma \ref{lemma:quadraticrenormalization}.
\begin{proof}[Proof of Lemma \ref{lemma:quadraticrenormalization}]
The computation is essentially the same as in \cite[Section 3]{BCS}; we include it for completeness due to the different choice of kernel. It can also be found in \cite{BOSS}. Using \eqref{eq:weyl1} we have 
    \begin{equation*}
        \la \Psi_N, \cH_N \Psi_N \ra
        = \sum_{n=0}^4 \la e^{B^*-B} \xi, \cW_n e^{B^*-B} \xi \ra
    \end{equation*}
     where 
    \begin{equation*}
        \begin{split}
            &\cW_0 = \frac{N^\kappa N_0^2}{2N} \hat{V}(0)\\
            &\cW_1 = \frac{N^\kappa N_0^{3/2}}{N} \hat{V}(0) a_0^*+\hc\\
            &\cW_2 = \cK 
                +\frac{N^\kappa N_0}{2N} \sum_{p\in\Lambda^*}
                \hat{V}(p/N^{1-\kappa}) \left( a_p^* a_{-p}^* + a_p a_{-p} + 2a_p^* a_{p} \right)
                +\frac{N^\kappa \hat{V}(0) N_0}{N} \sum_{p\in\Lambda^*}
                 a_p^* a_{p}\\
             &\cW_3 = \frac{N^\kappa N_0^{1/2}}{N} \sum_{p,r \in\Lambda^*} \hat{V}(r/N^{1-\kappa}) a_{p+r}^* a_{-r}^* a_p +\hc\\
            &\cW_4 = \cV_N.
        \end{split}
    \end{equation*}
    Write $\cG_n = e^{-B^* + B} \cW_n e^{B^*-B}$ for $n =0,\dots ,4$ for the conjugation of each term
      \begin{equation*}
        \la \Psi_N, \cH_N \Psi_N \ra
        = \sum_{n=0}^4 \la \xi, \cG_n  \xi \ra.
    \end{equation*}
    As $A^*$ only creates triples, in the following all expectation values of operators changing the particle number by neither $0$ nor $3$ vanish. Thus, 
    \begin{equation}\label{eq:G2xi} 
    \begin{split}
        \la \xi, \cG_2 \xi \ra
         =\; & \la \xi, \cK \xi \ra + \sum_{p\in\Lambda^*} p^2 \sigma_p^2 +  \frac{N^\kappa N_0}{N} \sum_{p\in\Lambda^*}
                \hat{V}(p/N^{1-\kappa}) (\gamma_p\sigma_p + \sigma_p^2) + \frac{N^\kappa \hat{V}(0) N_0}{N} \| \sigma \|_2^2\\
                &+ \langle\xi, \cQ_2\xi\rangle
    \end{split}
    \end{equation}
    where
    $\cQ_2$ is a quadratic operator satisfying
$
    \pm\cQ_2 \leq C N^{\kappa+\epsilon} \cN.
$
Similarly,
    \begin{equation*}\label{eq:G3xi} 
        \la\xi, \cG_3 \xi\ra = \frac{N^\kappa N_0^{1/2}}{N} \sum_{p, r\in\Lambda^*} \hat{V}(r/N^{1-\kappa}) 
        (\gamma_{p+r}\gamma_r\sigma_p + \sigma_{p+r}\sigma_r\gamma_p) 
        \la\xi, (a_{p+r}^* a_{-r}^* a_{-p}^* +\hc )\xi\ra.
    \end{equation*}
and
    \begin{equation}\label{eq:defG4}
    \begin{split}
        \la \xi, \cG_4 \xi \ra
        = \; &\frac{N^\kappa}{2N} \sum_{p,q,r\in\Lambda^*} \hat{V}(r/N^{1-\kappa}) 
        \big(\gamma_{p+r} \gamma_{q-r}\gamma_{q}\gamma_{p}
        +\sigma_{p+r}\sigma_{q-r}\sigma_{q}\sigma_{p} +2\gamma_{p+r}\sigma_{q-r}\sigma_{q}\gamma_p \big) \\ &\hspace{9cm} \times 
         \la \xi,  a_{p+r}^* a_{q-r}^* a_{q} a_{p}
        \xi\ra\\
        &+ \frac{N^\kappa}{N} \sum_{p,q,r\in\Lambda^*} \hat{V}(r/N^{1-\kappa}) 
        \gamma_{p+r}\sigma_{q-r}\gamma_{q}\sigma_{p}
         \la \xi, a_{p+r}^*a_{-p}^*  a_{-q+r} a_{q} 
        \xi \ra\\
        &+
         \frac{N^\kappa}{2N} \sum_{p,r\in\Lambda^*} \hat{V}(r/N^{1-\kappa})  
        (\sigma_{p+r}\gamma_{p+r}
        \gamma_{p}\sigma_{p}+
        \sigma_{p+r}^2\sigma_{p}^2)
        +
        \frac{N^\kappa}{2N} \hat{V}(0) 
        \norm{\sigma}_2^4
        + \langle\xi, \cQ_2\xi\rangle
    \end{split}
    \end{equation}
    with $
    \pm\cQ_4 \leq C N^{\kappa+\epsilon} \cN.
$
Thus, the only step left in the proof of Lemma \ref{lemma:quadraticrenormalization} is to analyse the constant terms.
Observe that $\langle\xi,  \cG_0 \xi \rangle$ together with the last constant on the r.h.s. of (\ref{eq:G2xi}) and the last constant on the r.h.s. of (\ref{eq:defG4}) gives
\[ 
\frac{N^\kappa \hat{V} (0)}{2N} (N_0^2+ 2 N_0\| \sigma \|_2^2 + \| \sigma \|_2^4) = \frac{N^\kappa}{2N} \hat{V} (0) \big(N_0 + \| \sigma \|_2^2 \big)^2 = \frac{N^{1+\kappa}}{2} \hat{V} (0) + O(N^{5\kappa/2-\epsilon}), \]
where we used \eqref{eq:particle_number_0}.
Additionally, the second last constant in \eqref{eq:defG4} is bounded by $C N^{\kappa-1} \| \hat{V} \|_\infty \| \sigma \|_2^4 \leq C N^{4\kappa - 1}. $
Furthermore, the first constant in \eqref{eq:defG4} can be decomposed as
\[ \begin{split} \frac{N^\kappa}{2N} \sum_{p,r} \hat{V} (r/N^{1-\kappa}) \sigma_{p+r} &\gamma_{p+r}\gamma_p \sigma_p  = \frac{N^\kappa}{2N} \sum_{p,r} \hat{V} (r/N^{1-\kappa}) \eta_{\infty,p+r} \eta_{\infty,p}  \\&+ \frac{N^\kappa}{N} \sum_{p,r} \hat{V} (r/N^{1-\kappa}) \eta_{\infty, p+r} \big( \gamma_{\infty,p} \sigma_{\infty,p} - \eta_{\infty,p} \big) + O(N^{5\kappa/2-\epsilon}) .\end{split} \]
Here we used $\norm{\sigma -\sigma_\infty}_1\leq N^{3\kappa/2-\epsilon}$ and $
\norm{\gamma_{\infty} \sigma_{\infty} - \eta_{\infty}}_1 \leq N^{3\kappa/2}$ from Lemma \ref{lemma:scattering}.

We treat the remaining constant terms in \eqref{eq:G2xi} analogously and obtain
\[ 
\begin{split} 
\frac{N^\kappa N_0}{N} \sum_p \hat{V} (p/N^{1-\kappa}) \gamma_p \sigma_p   
 &= N^\kappa \sum_p \hat{V} (p/N^{1-\kappa}) \eta_{\infty,p} - \frac{N^\kappa \| \sigma \|_2^2}{N} \sum_p \hat{V} (p/N^{1-\kappa}) \eta_{\infty,p} \\ &\hspace{.4cm} +  N^\kappa \sum_p \hat{V} (p/N^{1-\kappa}) (\gamma_{\infty,p} \sigma_{\infty,p} - \eta_{\infty,p}) + O(N^{5\kappa/2-\epsilon}) 
 \end{split}
 \]
and
\[ \begin{split} \frac{N^\kappa N_0}{N} &\sum_p \hat{V} (p/N^{1-\kappa}) \sigma_p^2  = N^\kappa \sum_p \hat{V} (p/N^{1-\kappa}) \sigma_{\infty,p}^2 + O(N^{5\kappa/2-\epsilon})
\end{split} \]
as well as
\[
\sum_p p^2 \sigma_p^2 = \sum_p p^2 \sigma_{\infty,p}^2 + O(N^{5\kappa/2-\epsilon}).
\]
Note that here we used the choice of $N_0$ in \eqref{eq:particle_number_0} again.
Thus, the non-negligible constants are
\begin{equation}\label{eq:const}
    \begin{split}
        &\frac{N^{1+\kappa}}{2} \hat{V} (0) + \sum_{p\in\Lambda^*} p^2 \sigma_{\infty,p}^2
        + N^\kappa\sum_{p\in\Lambda^*}
                \hat{V}(p/N^{1-\kappa}) \eta_{\infty,p} 
        - \frac{N^\kappa \norm{\sigma}_2^2}{N} \sum_{p\in\Lambda^*}
        \hat{V}(p/N^{1-\kappa}) \eta_{\infty,p}\\
        &+ N^\kappa \sum_{p\in\Lambda^*}
                \hat{V}(p/N^{1-\kappa}) (\gamma_{\infty,p}\sigma_{\infty,p}-\eta_{\infty,p})
        + N^\kappa \sum_{p\in\Lambda^*}
                \hat{V}(p/N^{1-\kappa})
                \sigma_{\infty,p}^2\\
        &+\frac{N^\kappa}{2N} \sum_{p,r\in\Lambda^*} \hat{V}(r/N^{1-\kappa})  
            \eta_{\infty,p+r}
            \eta_{\infty,p}          
            +\frac{N^\kappa}{N} \sum_{p,r\in\Lambda^*} \hat{V}(r/N^{1-\kappa}) 
            \eta_{\infty,p+r}
            (\gamma_{\infty,p}\sigma_{\infty,p}-\eta_{\infty,p}).\\
    \end{split}
\end{equation}
We have 
\[ \begin{split} N^\kappa\sum_{p\in\Lambda^*}
                \hat{V}(p/N^{1-\kappa}) \eta_{\infty,p} &+\frac{N^\kappa}{2N} \sum_{p,r\in\Lambda^*} \hat{V}(r/N^{1-\kappa})  
            \eta_{\infty,p+r}
            \eta_{\infty,p} \\ = &- \sum_{p \in \Lambda^*} p^2 \eta_{\infty,p}^2 
       + \frac{N^\kappa}{2} \sum_{p \in \Lambda^*} \hat{V} (p/N^{1-\kappa}) \eta_{\infty,p} + O(N^{2\kappa})\\ 
       = &- \sum_{p \in \Lambda^*} p^2 \eta_{\infty,p}^2  + \frac{8\pi \frak{a} - \hat{V} (0)}{2} N^{1+\kappa} + O(N^{2\kappa})\end{split} \]
thanks to the scattering equation \eqref{eq:scatteringdiscrete} as well as \eqref{eq:etacorrection}.     
With \eqref{eq:gammasigmainfty-etainftysum}, we find   
\[ \begin{split} 
 N^\kappa \sum_{p\in\Lambda^*}
                \hat{V}(p/N^{1-\kappa}) (\gamma_{\infty,p}\sigma_{\infty,p}-\eta_{\infty,p}) &+\frac{N^\kappa}{N} \sum_{p,r\in\Lambda^*} \hat{V}(r/N^{1-\kappa}) 
            \eta_{\infty,p+r}
            (\gamma_{\infty,p}\sigma_{\infty,p}-\eta_{\infty,p}) \\ &= -2 \sum_{p \in \Lambda^*} p^2 \eta_{\infty,p} ( \gamma_{\infty,p}\sigma_{\infty,p}-\eta_{\infty,p}) + O(N^{2\kappa}). \end{split} \]
Lastly, we decompose
\[ \begin{split} N^\kappa \sum_{p\in \Lambda^*} \hat{V} (p/N^{1-\kappa}) \sigma_{\infty,p}^2 = &-2 \sum_{p\in \Lambda^*} p^2 \eta_{\infty,p} \sigma_{\infty,p}^2 \\ &- N^{\kappa-1} \sum_{p,r \in \Lambda^*} \hat{V} (r/N^{1-\kappa}) \eta_{\infty,p+r} \sigma_{p}^2 + O(N^{2\kappa})\end{split} \]
allowing to rewrite \eqref{eq:const} as 
\[ \begin{split} 
&4\pi \frak{a} N^{1+\kappa} + \sum_{p\in\Lambda^*} p^2 (\sigma_{\infty,p}^2 + \eta_{\infty,p}^2 -2\eta_{\infty,p} \gamma_{\infty,p} \sigma_{\infty,p} -2 \eta_{\infty,p} \sigma_{\infty,p}^2)
       \\
        & - \frac{N^\kappa \norm{\sigma}_2^2}{N} \sum_{p\in\Lambda^*}
        \hat{V}(p/N^{1-\kappa}) \eta_{\infty,p} - N^{\kappa-1} \sum_{p,r\in \Lambda^*} \hat{V} (r/N^{1-\kappa}) \eta_{p+r} \sigma_{p}^2 + O(N^{2\kappa}).
                \end{split}
                \]
Plugging in \eqref{eq:sigmainfty}, we arrive at 
\[ \begin{split} 
&4\pi \frak{a} N^{1+\kappa}   + \frac{1}{2}\sum_{p\in\Lambda^*} p^2 
        \left(
        \left(\frac{1-2\eta_{\infty,p}}{\sqrt{1-4\eta_{\infty,p}}}-1\right) (1-2\eta_{\infty,p})
        + 2\eta_{\infty,p}^2
        - 4 \frac{\eta_{\infty,p}^2}{\sqrt{1-4\eta_{\infty,p}}}
        \right) \\
        &\hspace{.2cm}  - \frac{N^\kappa \norm{\sigma}_2^2}{N} \sum_{p\in\Lambda^*}
        \hat{V}(p/N^{1-\kappa}) \eta_{\infty,p} - N^{\kappa-1} \sum_{p,r\in \Lambda^*} \hat{V} (r/N^{1-\kappa}) \eta_{p+r} \sigma_{p}^2 + O(N^{2\kappa}) \\
        &= 4\pi \frak{a} N^{1+\kappa} + \frac{1}{2}\sum_{p\in\Lambda^*} p^2 
        \left(
        \sqrt{1-4\eta_{\infty,p}}
        - 1 + 2\eta_{\infty,p}
        + 2\eta_{\infty,p}^2
        \right) \\
          &\hspace{.2cm}  - \frac{N^\kappa \norm{\sigma}_2^2}{N} \sum_{p\in\Lambda^*}
        \hat{V}(p/N^{1-\kappa}) \eta_{\infty,p} - N^{\kappa-1} \sum_{p,r\in \Lambda^*} \hat{V} (r/N^{1-\kappa}) \eta_{p+r} \sigma_{p}^2  + O(N^{2\kappa})
                \end{split} \]
                showing the claim after inserting the definition
\eqref{eq:defetainfty} for $\eta_{\infty,p}$.
\end{proof}

\end{document}